\documentclass[twoside,english,onecolumn,5p]{elsarticle}
\usepackage[T1]{fontenc}
\usepackage[latin9]{inputenc}
\usepackage{color}
\usepackage{float}
\usepackage{dsfont}
\usepackage{amsmath}
\usepackage{amsthm}
\usepackage{amssymb}
\usepackage{graphicx}

\makeatletter

\floatstyle{ruled}
\newfloat{algorithm}{tbp}{loa}
\providecommand{\algorithmname}{Algorithm}
\floatname{algorithm}{\protect\algorithmname}

\theoremstyle{definition}
\newtheorem{defn}{\protect\definitionname}
\theoremstyle{plain}
\newtheorem{lem}{\protect\lemmaname}
\theoremstyle{remark}
\newtheorem{rem}{\protect\remarkname}
\theoremstyle{definition}
\newtheorem{example}{\protect\examplename}
\theoremstyle{plain}
\newtheorem{thm}{\protect\theoremname}
\ifx\proof\undefined\
  \newenvironment{proof}[1][\proofname]{\par
    \normalfont\topsep6\p@\@plus6\p@\relax
    \trivlist
    \itemindent\parindent
    \item[\hskip\labelsep
          \scshape
      #1]\ignorespaces
  }{%
    \endtrivlist\@endpefalse
  }
  \providecommand{\proofname}{Proof}
\fi

\journal{}

\usepackage[caption=false,font=footnotesize]{subfig}

\date{}
\usepackage{amstext}
\usepackage{pdfsync}
\usepackage{amsthm}

\usepackage{epstopdf}

\usepackage{tikz}
\usetikzlibrary{arrows}
\usetikzlibrary{shadows}
\tikzset{
  every overlay node/.style={
    draw=white,anchor=north west,
  },
}
\usetikzlibrary{positioning, decorations.pathmorphing}

\usetikzlibrary{fit,positioning,hobby}

\usepackage{hyperref}
\@ifundefined{rangeHsb}{\usepackage{xcolor}}{}

\usepackage{algorithm,algpseudocode}

\usepackage{lipsum}

\usepackage{scrlayer-scrpage}
\clearpairofpagestyles
\makeatother

\usepackage{babel}
\providecommand{\definitionname}{Definition}
\providecommand{\examplename}{Example}
\providecommand{\lemmaname}{Lemma}
\providecommand{\remarkname}{Remark}
\providecommand{\theoremname}{Theorem}

\begin{document}

\begin{frontmatter}{}

\title{Vector-valued Privacy-Preserving  Average Consensus}

\author[ra]{Lulu~Pan}

\author[ra]{Haibin~Shao \corref{cor1}}

\author[rd]{Yang~Lu}

\author[focal]{Mehran~Mesbahi}

\author[ra]{Dewei~Li}

\author[ra]{Yugeng~Xi}

\fntext[fn2]{The research of H. Shao, L. Pan, D. Li, and Y. Xi has been supported
by the National Natural Science Foundation of China (Grant No. 62103278,
61973214,  61963030). }

\cortext[cor1]{Corresponding author: Haibin Shao (shore@sjtu.edu.cn).}

\address[ra]{Department of Automation, Shanghai Jiao Tong University, Shanghai,
200240, China}

\address[rd]{School of Computing and Communications, Lancaster University, Lancaster,
LA1 4WA, United Kingdom}

\address[focal]{William E. Boeing Department of Aeronautics and Astronautics, University
of Washington, Seattle, WA 98195-2400, USA}
\begin{abstract}
Achieving average consensus without disclosing sensitive information
can be  a critical concern for multi-agent coordination. This paper
examines privacy-preserving average consensus (PPAC) for vector-valued
multi-agent networks. In particular, a set of agents with vector-valued
states aim to collaboratively reach an exact average consensus of
their initial states, while each agent\textquoteright s initial state
cannot be disclosed to other agents. We show that the vector-valued
PPAC problem can be solved via associated matrix-weighted networks
with the higher-dimensional agent state. Specifically, a novel distributed
vector-valued PPAC algorithm is proposed by lifting the agent-state
to higher-dimensional space and designing the associated matrix-weighted
network with dynamic, low-rank, positive semi-definite coupling matrices
to both conceal the vector-valued agent state and guarantee that the
multi-agent network asymptotically converges to the average consensus.
Essentially, the convergence analysis can be transformed into the
average consensus problem on switching matrix-weighted networks. We
show that the exact average consensus can be guaranteed and the initial
agents' states can be kept private if each agent has at least one
``legitimate'' neighbor. The algorithm, involving only basic matrix
operations, is computationally more efficient than cryptography-based
approaches and can be implemented in a fully distributed manner without
relying on a third party. Numerical simulation is provided to illustrate
the effectiveness of the proposed algorithm.

\,
\end{abstract}

\end{frontmatter}{}

\section{Introduction}

Distributed average consensus plays a crucial role in distributed
estimation, control, and optimization of networked multi-agent systems
\citet{mesbahi2010graph,kia2019tutorial}. Specifically, a set of
agents iteratively exchange information with their neighbors over
a communication graph such that, under certain conditions on the inter-agent
couplings, all the agents' states converge to the average of their
initial states. Traditional average consensus algorithms admit the
exchange of explicit states between neighboring agents, which may
disclose initial states of the agents to malicious ones. In certain
average consensus applications, however, an agent's initial state
contains sensitive information and must be kept confidential. One
example is opinion consensus in social networks, where a set of agents
holding respective opinions leverage consensus algorithms to reach
a common opinion but are unwilling to reveal their own opinions as
these opinions may reflect sensitive personal preferences \citet{mo2016privacy}.
Another example is dynamic formation control, where a set of mobile
agents employ consensus algorithms to agree on, say, the time-varying
geometric center of their formation \citet{PORFIRI20071318}. However,
these agents may not want to expose their own locations during this
process, as that could reveal their habits, interests, activities,
and relationships \citet{5601955}. This necessitates new class of
algorithms that can simultaneously achieve average consensus while
protecting the privacy of individual agents' initial states.

\subsection*{Literature Review }

The problem of privacy-preserving average consensus (PPAC) has been
extensively examined in recent literature. Perturbing the original
data prior to sharing it on the network is a typical option for privacy-enhancing
coordination \citet{mo2016privacy,he2018distributed,6669251}. Here,
differential privacy is a formal framework for the design and analysis
of privacy-preserving algorithms by persistently adding random noises
into agents' states such that, with high probability, the initial
state of an agent cannot be inferred by adversarial agents \citet{8486684,NOZARI2017221,he2020differential}.
However, there is a fundamental trade-off between privacy and accuracy
of the consensus value as persistent random noises are incorporated
\citet{he2020differential,cortes2016differential,lu2019control}.
Notably, some of the works have proposed utilizing decaying or correlated
noises to obfuscate exchanged states such that individual agents'
initial states cannot be uniquely determined while average consensus
can be maintained \citet{6669251,8619133,he2018distributed}. However,
there is a risk of revealing an interval of an agent's initial state,
corresponding to the magnitude of the added noises \citet{he2018distributed}.
Cryptography-based approaches are also employed for privacy-preserving
algorithm design. In fact, the PPAC problem is closely related to
secure multi-party computation, where all participants compute a joint
function while preserving their respective inputs private \citet{lu2019control,paillier1999public,lu2020privacy,6375935,KK-TF:2015}.
In this line of work, additive properties of homomorphic encryption
are employed to encrypt exchanged states such that the desired computations
can be carried out on encrypted states, thus generating an encrypted
result that, when decrypted, matches the result of computations performed
on original states \citet{ruan2019secure}. Nevertheless, cryptography-based
approaches may suffer from communication and computation overhead
caused by the encryption process; moreover a centralized authority
is necessary to carry out aggregation over encrypted data \citet{lu2019control,chong2019tutorial}.
In \citet{wang2019privacy}, a state-decomposition-based protocol
has been proposed where the state of a node is randomly decomposed
into two substates such that the mean of these states remains the
same and only one of the substates is revealed to neighboring nodes.
Privacy-preserving problems for continuous-time dynamical systems
have also been examined recently by \citet{altafini2020system,xiong2022privacy}.

To the best of our knowledge, most existing privacy-preserving average
consensus algorithms are specifically designed for scalar-valued agent
states; however, vector-valued information exchange is ubiquitous
in multi-agent networks. For instance, vector-valued local estimates
of the optimal solution are exchanged amongst neighboring agents for
multi-agent optimization \citet{yang2019survey}; point vectors of
neighboring agents in a global reference frame are exchanged to construct
relative bearing vector in the distributed control law for bearing-based
multi-agent formation \citet{zhao2015bearing,trinh2018bearing}. 

Inspired by these recent developments in matrix-weighted networks,
where the couplings amongst vector-valued agent states are characterized
by square matrices \citet{trinh2018matrix,tuna2016synchronization,pan2018bipartite,wang2022characterizing,pan2021consensus,pan2021cluster,pan2020controllability},
this paper proposes a novel distributed vector-valued PPAC algorithm
by employing matrix-valued state coupling amongst the agents. In this
view, the state coupling of neighboring agents in traditional scalar-weighted
networks can be regarded as a $d$-order identity matrix where $d$
refers to the dimension of the agents' states. We show that vector-valued
PPAC can in fact be transformed into a properly constructed matrix-weighted
networks.

Contributions. In this paper, a novel distributed vector-valued PPAC
algorithm, based on matrix-valued state coupling, is proposed. The
idea of the algorithm is to utilize properly constructed, low-rank,
semi-definite coupling matrices, together with agent-state lifting
(introducing virtual states), to conceal the vector-valued agent states,
and in the meantime, employ a periodic edge weight switching mechanism
to guarantee that the multi-agent network asymptotically converges
to the exact average consensus. We show that the exact average consensus
can be guaranteed and the initial agents' states can be preserved
by the proposed algorithm  if each agent has at least one ``legitimate''
neighbor. Since our algorithm only involves basic matrix multiplication,
it is more computationally efficient than cryptography-based algorithms.
Moreover, the algorithm can be implemented in a distributed manner
without the involvement of a third party.

The results in this paper have the following immediate implications.
Firstly, the vector-valued PPAC framework can be immediately applied
in distributed control and optimization algorithms with privacy-preservation
concerns \citet{dibaji2019systems,kia2019tutorial}. Secondly, we
present general results on the average asymptotic consensus problems
on time-varying discrete-time matrix-weighted networks, which is of
independent interest in distributed control and optimization \citet{trinh2018matrix,wang2022characterizing,tuna2016synchronization}.
Moreover, the proposed algorithm can be applied to scaler-valued PPAC
problem, where the virtual states can serve as an effective means
of privacy-preservation.

The remainder of the paper is organized as follows. Preliminaries
and problem formulation are presented in $\mathsection$\ref{sec:problem-formulation},
where notation and a brief introduction to the matrix-weighted network
are introduced. Vector-valued PPAC algorithm, based on matrix-valued
inter-agent state coupling, is proposed in $\mathsection$\ref{sec:Algorithm},
followed by the analysis of average consensus and privacy-preserving
performance in $\mathsection$\ref{sec:Consensus} and $\mathsection$\ref{sec:Privacy},
respectively. Simulation results are presented in $\mathsection$\ref{sec:Simulation};
we provide concluding remarks in $\mathsection$\ref{sec:Conclusion}.

\section{Notation and Problem Formulation\label{sec:problem-formulation}}

We first introduce the notation. Let $\mathbb{R}$, $\mathbb{N}$
and $\mathbb{Z}_{+}$ be the set of real numbers, natural numbers
and positive integers, respectively. For $n\in\mathbb{Z}_{+}$, denote
$\underline{n}=\left\{ 1,2,\ldots,n\right\} $. We use $M\succ0$
(respectively, $M\succeq0$) to denote that a symmetric matrix $M$
is positive definite (respectively, positive semi-definite). The null
space and range space of a matrix $M$ is denoted by $\text{{\bf null}}(M)$
and $\text{{\bf range}}(M)$, respectively. Let $\mathds{1}_{n}$,
$\boldsymbol{0}_{n\times n}$ and $\boldsymbol{I}_{n}$ designate
the $n$-dimensional column vector whose components are all $1$'s,
the $n\times n$ matrix whose components are all $0$'s, and the $n\times n$
identity matrix, respectively. The eigenspace of $M\in\mathbb{R}^{n\times n}$
corresponding to eigenvalue $\lambda\in\mathbb{R}$ is denoted by
$\mathbb{E}_{\lambda}=\left\{ \boldsymbol{v}\in\mathbb{R}^{n}\thinspace|\thinspace M\boldsymbol{v}=\lambda\boldsymbol{v}\right\} $.
We use $a\thinspace\text{{\bf mod}}\thinspace b$ to refer to the
remainder of the Euclidean division where $a\in\mathbb{Z}_{+}$ is
the dividend and $b\in\mathbb{Z}_{+}$ is the divisor. Let $\delta(x):\mathbb{R}\mapsto\left\{ 0,1\right\} $
denote the sign function such that $\delta(x)=1$ for $x>0$ and $\delta(x)=0$
for $x\le0.$ For a vector $\boldsymbol{x}\in\mathbb{R}^{n}$ and
$i_{1}<i_{2}\in\mathbb{Z}_{+}$, we write $\boldsymbol{x}^{[i1:i2]}$
to denote a $\mathbb{R}^{(i_{2}-i_{1}+1)}$ vector formed by entries
of $\boldsymbol{x}$ in sequence from $i_{1}$ to $i_{2}$. 

Consider a multi-agent system consisting of $n>1$ ($n\in\mathbb{Z}_{+}$)
agents whose interaction network is characterized by a communication
graph $\mathcal{G}=(\mathcal{V},\mathcal{E},A)$. The node and edge
sets of $\mathcal{G}$ are denoted by $\mathcal{V}=\left\{ 1,2,\ldots,n\right\} $
and $\mathcal{E}\subseteq\mathcal{V}\times\mathcal{V}$, respectively.
A path in $\mathcal{G}$ is a sequence of edges of the form $(i_{1},i_{2}),(i_{2},i_{3}),\ldots,(i_{p-1},i_{p})$,
where nodes $i_{1},i_{2},\ldots,i_{p}\in\mathcal{V}$ are distinct;
in this case we say that node $i_{p}$ is reachable from $i_{1}$.
A graph $\mathcal{G}$ is connected if any two distinct nodes in $\mathcal{G}$
are reachable from each other. A tree is a connected graph with $n\ge2$
nodes and $n-1$ edges where $n\in\mathbb{Z}_{+}$. For matrix-weighted
switching networks, we adopt the following terminology. An edge $(i,j)\in\mathcal{E}$
is positive definite (semi-definite) if $A_{ij}$ is positive definite
(semi-definite). A positive tree of $\mathcal{G}$ is a tree such
that every edge in this tree is positive definite. A positive spanning
tree of $\mathcal{G}$ is a positive tree containing all nodes in
$\mathcal{G}$. 

\subsection{Matrix-weighted Average Consensus}

In a matrix-weighted network $\mathcal{G}=(\mathcal{V},\mathcal{E},A)$,
each edge $(i,j)\in\mathcal{E}$ is assigned a matrix-valued weight
encoded by a matrix $A_{ij}\in\mathbb{R}^{d\times d}$ such that $A_{ij}\not=\boldsymbol{0}_{d\times d}$
if $(i,j)\in\mathcal{E}$, and $A_{ij}=\boldsymbol{0}_{d\times d}$
otherwise. We shall assume that all non-zero matrix-valued weights
are either positive definite or positive semi-definite unless otherwise
stated. Thereby, the matrix-valued adjacency matrix $A=(A_{ij})\in\mathbb{R}^{dn\times dn}$
is a block matrix such that the block located in its $i$-th row and
the $j$-th column is $A_{ij}$. Each agent $i$ has a vector-valued
initial state $\boldsymbol{x}_{i}(0)=(x_{i1}(0),\ldots,x_{id}(0))^{\top}\in\mathbb{R}^{d}$.
The agents aim to agree on the average of their initial states. To
this end, each agent $i$ updates its state by the protocol, 
\begin{equation}
\boldsymbol{x}_{i}(k+1)=\boldsymbol{x}_{i}(k)+\sigma\sum_{j\in\mathcal{N}_{i}}A_{ij}(\boldsymbol{x}_{j}(k)-\boldsymbol{x}_{i}(k)),\label{equ:matrix-consensus-protocol}
\end{equation}
where $\sigma>0$, $k\in\mathbb{N}$ is the step index, and 
\[
\mathcal{N}_{i}=\left\{ j\in\mathcal{V}\,|\,(i,j)\in\mathcal{E}\right\} 
\]
denotes the neighbor set of agent $i$. The matrix-valued Laplacian
matrix of $\mathcal{G}$ is defined as $L=D-A$, where $D=\text{{\bf diag}}\left\{ D_{1},\cdots,D_{n}\right\} \in\mathbb{R}^{dn\times dn}$
and ${\color{black}{\color{black}{\color{red}{\color{blue}{\color{black}D_{i}=\sum_{j\in\mathcal{N}_{i}}A_{ij}\in\mathbb{R}^{d\times d}}}}}}$.
Subsequently, the overall dynamics of \eqref{equ:matrix-consensus-protocol}
can be compactly written as, 
\begin{equation}
\boldsymbol{x}(k+1)=(\boldsymbol{I}_{dn}-\sigma L)\boldsymbol{x}(k),\label{equ:matrix-consensus-overall}
\end{equation}
where $\boldsymbol{x}(k)=(\boldsymbol{x}_{1}^{\top}(k),\ldots,\boldsymbol{x}_{n}^{\top}(k))^{\top}\in\mathbb{R}^{dn}$.

The definition of asymptotic average consensus is formalized next.
\begin{defn}
The matrix-weighted network \eqref{equ:matrix-consensus-overall}
achieves asymptotic average consensus if 
\begin{equation}
{\color{black}{\color{blue}{\color{black}{\displaystyle \lim_{k\rightarrow\infty}}\boldsymbol{x}(k)=\mathds{1}_{n}\otimes\left(\text{Avg}(\boldsymbol{x}(0))\right)}},}
\end{equation}
where $\text{Avg}(\boldsymbol{x}(0))=\frac{1}{n}\sum_{i=1}^{n}\boldsymbol{x}_{i}(0)$
denotes the average value of all the $n$ agents' initial states. 
\end{defn}
\textcolor{black}{The following lemmas characterizes the null space
of matrix-weighted Laplacian which plays a central role in achieving
asymptotic average consensus under \eqref{equ:matrix-consensus-overall}.}
\begin{lem}
\label{lem:1-1}\citet{pan2018bipartite,trinh2018matrix} Let $\mathcal{G}=(\mathcal{V},\mathcal{E},A)$
be a matrix-weighted network. Then the associated matrix-valued Laplacian
matrix $L$ of $\mathcal{G}$ is positive semi-definite and the structure
of its null space can be characterized by $\text{{\bf null}}(L)=\text{{\bf span}}\left\{ \mathcal{R},\mathcal{H}\right\} ,$
where 
\begin{equation}
\mathcal{R}=\text{{\bf range}}\{\mathds{1}_{n}\otimes I_{d}\},\label{eq:consensus-space}
\end{equation}
and 
\begin{align}
\mathcal{H}=\{\boldsymbol{v} & =(\boldsymbol{v}_{1}^{\top},\boldsymbol{v}_{2}^{\top},\cdots,\boldsymbol{v}_{n}^{\top})^{\top}\in\mathbb{R}^{dn}\mid\nonumber \\
 & (\boldsymbol{v}_{i}-\boldsymbol{v}_{j})\in\text{{\bf null}}(A_{ij}),\,(i,j)\in\mathcal{E}\}.\label{eq:edge-space}
\end{align}
\end{lem}
\begin{lem}
\label{lem:1-2}\citet{trinh2018matrix,pan2018bipartite} Let $\mathcal{G}=(\mathcal{V},\mathcal{E},A)$
be a matrix-weighted network\textcolor{red}{.} If $\mathcal{G}$ has
a positive spanning tree\footnote{A positive spanning tree in a matrix-weighted network $\mathcal{G}=(\mathcal{V},\mathcal{E},A)$
is a spanning tree $\mathcal{T}$ of $\mathcal{G}$ such all edge
weights in $\mathcal{T}$ are positive definite matrices.}, then the matrix-valued Laplacian $L$ satisfies $\text{{\bf null}}(L)=\mathcal{R}$. 
\end{lem}
\textcolor{black}{A notable feature of continous-time multi-agent
networks on scalar-weighted networks is that network connectivity
can translate into achieving consensus, which is not valid for matrix-weighted
networks. In this case, it is therefore intricate to obtain a purely
graph-theoretic condition for achieving consensus without any assumptions
on the matrix-valued edge weights, the null space of which play a
paramount role in achieving consensus on continous-time matrix-weighted
networks \citet{trinh2018matrix,pan2018bipartite}. }

\textcolor{black}{Specifically, according to Lemma \ref{lem:1-1},
the null space of a matrix-valued Laplacian is not only determined
by the network connectivity, but also by the null space of each matrix-valued
edge weight. In this case, the null space of the matrix-valued Laplacian
associated with a connected matrix-weighted network may not be $\mathcal{R}$,
implying that the discrete-time multi-agent system \eqref{equ:matrix-consensus-overall}
on connected matrix-weighted networks may also achieve cluster consensus
due to the complexity of $\text{{\bf null}}(L)$, regardless of the
selection of $\sigma$. We have the similar observations; see Example
\ref{exa:cluster-consensus}. }

\subsection{Consensus with Privacy}

For preserving privacy, it is not preferable to transmit agents' states
in \eqref{equ:matrix-consensus-protocol}. Instead, by a rearrangement
of \eqref{equ:matrix-consensus-protocol},
\begin{align}
\boldsymbol{x}_{i}(k+1) & =\boldsymbol{x}_{i}(k)+\sigma\sum_{j\in\mathcal{N}_{i}}(A_{ij}\boldsymbol{x}_{j}(k)-A_{ij}\boldsymbol{x}_{i}(k)),\label{eq:protocol-rearrange}
\end{align}
one can immediately see that each $j\in\mathcal{N}_{i}$ only needs
to send
\begin{equation}
\boldsymbol{y}_{j\rightarrow i}(k)=A_{ij}\boldsymbol{x}_{j}(k)
\end{equation}
to $i$ for \textcolor{black}{its} state update. Therefore, the state
$\boldsymbol{x}_{j}(k)$ can be concealed by a properly designed weight
matrix $A_{ij}$, yielding the following protocol,

\begin{align}
\boldsymbol{x}_{i}(k+1) & =\boldsymbol{x}_{i}(k)+\sigma\sum_{j\in\mathcal{N}_{i}}(\boldsymbol{y}_{j\rightarrow i}(k)-A_{ij}\boldsymbol{x}_{i}(k)).\label{eq:protocol-transmission}
\end{align}
Note that the weight matrix $A_{ij}$ cannot be positive definite,
otherwise, agent $i$ can immediately infer the state of agent $j$
by $A_{ij}^{-1}(A_{ij}\boldsymbol{x}_{j}(k))=\boldsymbol{x}_{j}(k)$.
On the other hand, according to Lemma \ref{lem:1-2}, the design of
each $A_{ij}$ also has to guarantee the existence of positive spanning
tree in $\mathcal{G}$ for average consensus. The aforementioned two
criteria for edge weight matries can conflict with each other. This
motivates us to consider dynamic matrix-valued edge weight mechanism,
examined in $\mathsection$ \ref{subsec:Dynamic-Matrix-valued-Edge}.

\subsection{Attacker Model and Privacy Definition}

This paper is concerned with the honest-but-curious (or semi-honest)
attacker model, i.e., an adversarial agent follows the designed algorithm
but attempts to use its received data to infer other agents' private
data (\citet{Hazay}). We shall refer to this type of adversarial
agent as an honest-but-curious agent. This attacker model has been
widely used in PPAC \citet{ruan2019secure,he2018distributed,huang2012differentially}.
Moreover, a legitimate agent is an agent who follows the consensus
protocol faithfully without attempting to infer other agents\textquoteright{}
states.

To formalize privacy, we introduce the following terms. Let $\mathcal{V}_{\mathcal{B}}$
and $\mathcal{V}_{\mathcal{A}}$ be the set of legitimate (or benign)
and honest-but-curious agents, respectively. Notice that $\mathcal{V}=\mathcal{V}_{\mathcal{B}}\cup\mathcal{V}_{\mathcal{A}}$.
Let $\mathcal{M}$ be a privacy-preserving algorithm \eqref{eq:protocol-rearrange}.
For each $i\in\mathcal{V}$, let $\mathcal{I}_{i}$ be agent $i$'s
input other than $\boldsymbol{x}_{i}(0)$ to $\mathcal{M}$. For a
given input $\{\boldsymbol{x}_{i}(0),\mathcal{I}_{i}\}_{i\in\mathcal{V}}$,
denote by $\mathcal{O}_{\mathcal{S}}^{\mathcal{M}}(\{\boldsymbol{x}_{i}(0),\mathcal{I}_{i}\}_{i\in\mathcal{V}})$
the set of all the observations a subset of agents $\mathcal{S}\subset\mathcal{V}$
($\mathcal{S}\not=\emptyset$) can gain throughout the execution of
$\mathcal{M}$ over the input $\{\boldsymbol{x}_{i}(0),\mathcal{I}_{i}\}_{i\in\mathcal{V}}$.
We are now ready to formalize the notation of privacy \citet{wang2019privacy,rezazadeh2018privacy}.
\begin{defn}
\label{def:definition of privacy} Let $\mathcal{M}$ be a privacy-preserving
algorithm \eqref{eq:protocol-rearrange}. We say the privacy of the
initial value of agent $b\in\mathcal{V}_{\mathcal{B}}$, denoted by
$\boldsymbol{x}_{b}(0)$, is preserved if for any $\bar{\boldsymbol{x}}_{b}(0)\in\mathbb{R}^{d}$,
there exist associated $\bar{\mathcal{I}}_{b}$ and $\{\bar{\boldsymbol{x}}_{l}(0),\bar{\mathcal{I}}_{l}\}_{l\in\mathcal{V}_{\mathcal{B}}\backslash\{b\}}$
such that 
\begin{align}
 & \mathcal{O}_{\mathcal{V}_{\mathcal{A}}}^{\mathcal{M}}\left(\{\bar{\boldsymbol{x}}_{b}(0),\bar{\mathcal{I}}_{b}\},\{\bar{\boldsymbol{x}}_{l}(0),\bar{\mathcal{I}}_{l}\}_{l\in\mathcal{V}_{\mathcal{B}}\backslash\{b\}},\{\boldsymbol{x}_{a}(0),\mathcal{I}_{a}\}_{a\in\mathcal{V}_{\mathcal{A}}}\right)\nonumber \\
 & =\mathcal{O}_{\mathcal{V}_{\mathcal{A}}}^{\mathcal{M}}\left(\{\boldsymbol{x}_{i}(0),\mathcal{I}_{i}\}_{i\in\mathcal{V}}\right).\label{eq:privacy-definition}
\end{align}
\end{defn}
\begin{rem}
The intuition behind Definition \ref{def:definition of privacy} is
that, for any input $\{\boldsymbol{x}_{i}(0),\mathcal{I}_{i}\}_{i\in\mathcal{V}}$
and all $b\in\mathcal{V}_{\mathcal{B}}$, any $\bar{\boldsymbol{x}}_{b}(0)\in\mathbb{R}^{d}$
is admissible to the adversarial agents' observations gained through
the implementation of the privacy-preserving algorithm; hence any
vector in $\mathbb{R}^{d}$ is indistinguishable from the true initial
state $\boldsymbol{x}_{b}(0)$ to the adversarial agents \citet{wang2019privacy,rezazadeh2018privacy}.
\end{rem}
\begin{rem}
Definition \ref{def:definition of privacy} also implies that even
a bound of a legitimate agent's initial state cannot be determined
by adversarial agents. 
\end{rem}

\subsection{\label{subsec:Objectives}Objectives}

This paper aims to develop a vector-valued PPAC algorithm such that
the following properties are guaranteed simultaneously: 
\begin{itemize}
\item Exact average consensus: For each agent $i\in\mathcal{V}$, its state
asymptotically converges to $\text{Avg}(\boldsymbol{x}(0))$.
\item Privacy: The privacy of the initial states of legitimate agents is
protected in the sense of Definition \ref{def:definition of privacy}. 
\end{itemize}

\section{Algorithm Design \label{sec:Algorithm}}

This work intends to design a novel vector-valued PPAC algorithm by
employing the framework of matrix-weighted networks (MWNs).

We make the following assumption on the communication networks throughout
this paper.

\textbf{Assumption 1}: The communication graph $\mathcal{G}$ is undirected
and connected.

In the following, we shall refer to our algorithm as the MWN-PPAC
algorithm; let us first motivate its structure.

We denote by $\mathbb{\mathcal{A}}=\left\{ A_{ij}\thinspace|\thinspace(i,j)\in\mathcal{E}\right\} $
the profile of matrix-valued edge weights of $\mathcal{G}$. Then
profile $\mathbb{\mathcal{A}}$ needs to ensure the objectives in
$\mathsection$\ref{subsec:Objectives}. Essentially, our MWN-PPAC
algorithm is to achieve a fundamental  tradeoff between privacy preservation
and average consensus in designing profile $\mathcal{A}$. In fact,
the positive semi-definite edge weight matrices are plausible for
privacy-preservation purposes but not for the average consensus, since
according to Lemma \ref{lem:1-2}, positive definite edge weight matrices
are more preferable to guarantee the existence of a positive spanning
tree in matrix-weight networks. 

We shall now proceed to achieve the aforementioned design objectives,
that of privacy and consensus, by agent state lifting and dynamic
matrix-valued weights. 

\subsection{Agent State Lifting \label{subsec:Agent-State-Augmentation}}

First, in designing profile $\mathbb{\mathcal{A}}$, the rank of each
$A_{ij}$ needs to be as low as possible for privacy preservation.
Then, less information with respect to the correlation of entries
in $\boldsymbol{x}_{j}(k)$ can be inferred from $\boldsymbol{y}_{j\rightarrow i}(k)$.\textcolor{black}{{}
Note that designing low-rank matrix-weighted edge weights is the most
preferable for this purpose.} However, although one cannot infer each
element in $\boldsymbol{x}_{i}(0)$, a linear correlation amongst
the entries in $\boldsymbol{x}_{i}(0)$ can be deduced. To this end,
the algorithm is designed to first lift the individual state space
by introducing a $d^{\prime}$-dimensional virtual state for each
agent. The idea of this design is essential to sacrifice the possible
disclosure of linear correlation of entries in virtual states instead
of that in real agents' states. We shall technically demonstrate this
point in the proof of Theorem \ref{thm:Privacy-preserving}. For a
glimpse, here we consider the following example.
\begin{figure}[H]
\begin{centering}
\begin{tikzpicture}[scale=0.8]
	\node (n1) at (-3,0.3) [circle,circular drop shadow,fill=black!20,draw] {1};
    \node (n2) at (-2,-1.5) [circle,circular drop shadow,fill=black!20,draw] {2};
    \node (n3) at (-1 ,0.3) [circle,circular drop shadow,fill=black!20,draw] {3};
    \node (n4) at (1 ,0.3) [circle,circular drop shadow,fill=black!20,draw] {4};
	\node (n5) at (0,-1.5) [circle,circular drop shadow,fill=black!20,draw] {5};
    \node (G2) at (-1,-2.5) {{$\mathcal{G}$}};
	\draw[-, ultra thick, color=black!70] (n1) -- (n2); 
	\draw[-, ultra thick, color=black!70] (n2) -- (n3); 
    \draw[-, ultra thick, color=black!70] (n1) -- (n3); 
    \draw[-, ultra thick, color=black!70] (n2) -- (n5);
    \draw[-, ultra thick, color=black!70] (n3) -- (n4);
    \draw[-, ultra thick, color=black!70] (n4) -- (n5); 
    \draw[-, ultra thick, color=black!70] (n1) -- (n2); 
\end{tikzpicture}
\par\end{centering}
\caption{A $5$-node connected matrix-weighted network $\mathcal{G}$.}
\label{fig:integral-network}
\end{figure}
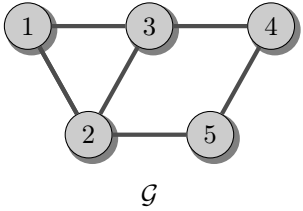

\begin{example}
\label{exa:why-switching-network}Consider a connected matrix-weighted
network $\mathcal{G}$ with $n=5$ agents in Figure \ref{fig:integral-network}.
Choose the following semi-definite matrix as edge weight at time $k=0$,
namely,
\[
A_{ij}=\left(\begin{array}{ccc}
\frac{1}{4} & -\frac{1}{4} & \frac{1}{2}\\
-\frac{1}{4} & \frac{1}{4} & -\frac{1}{2}\\
\frac{1}{2} & -\frac{1}{2} & 1
\end{array}\right),\thinspace(i,j)\in\mathcal{E}.
\]
Then the following linear correlation amongst the entries in $\boldsymbol{x}_{j}(0)$
can be inferred by agent $i$,
\begin{align}
A_{ij}\boldsymbol{x}_{j}(0) & =\frac{1}{4}x_{j1}(0)-\frac{1}{4}x_{j2}(0)+\frac{1}{2}x_{j3}(0)\nonumber \\
 & =\boldsymbol{y}_{j\rightarrow i}(0).\label{eq:correlation-x0}
\end{align}
\end{example}
Therefore, the agent state lifting procedure is employed in our algorithm
design to avoid disclosing the linear correlations amongst the entries
of real agents' states; we summarize this below.

Step 1: Agent State Lifting. 

Each agent $i\in\mathcal{V}$ lifts the state space by introducing
a $d^{\prime}$-dimensional virtual state,
\begin{equation}
\boldsymbol{x}_{i}^{v}=\left(x_{i1}^{v}(k),\ldots,x_{id^{\prime}}^{v}(k)\right)^{\top}\in\mathbb{R}^{d^{\prime}}.
\end{equation}

That is, the state of each agent $i\in\mathcal{V}$ is lifted from
\begin{equation}
\boldsymbol{x}_{i}(k)=(x_{i1}(k),\ldots,x_{id}(k))^{\top}\in\mathbb{R}^{d}
\end{equation}
into
\begin{equation}
\tilde{\boldsymbol{x}}_{i}(k)=(\tilde{x}_{i1}(k),\,\tilde{x}_{i2}(k),\cdots,\,\tilde{x}_{i,d+d^{\prime}}(k))^{\top}\in\mathbb{R}^{d+d^{\prime}},\label{eq:state-augmentation}
\end{equation}
where $\tilde{x}_{i,l_{1}}(k)=x_{il_{1}}^{v}(k)$ for $l_{1}\in\underline{d^{\prime}}$
and $\tilde{x}_{i,d^{\prime}+l_{2}}(k)=x_{il_{2}}(k)$ for $l_{2}\in\underline{d}$. 
\begin{rem}
The virtual state mechanism not only renders the vector-valued PPAC
algorithm proposed in this paper applicable to scalar-weighted networks,
but provides freedom to protect the correlation of entries in the
initial state of each agent from being disclosed.
\end{rem}
For simplicity, in the following discussion, we shall use the symbol
$\boldsymbol{x}_{i}(k)$ to denote the lifted agent $\tilde{\boldsymbol{x}}_{i}(k)$.
Therefore, the virtual state and real state of an agent $i$ in the
lifted agent state vector $\boldsymbol{x}_{i}(k)$ can be represented
by $\boldsymbol{x}_{i}^{[1:d^{\prime}]}(k)$ and $\boldsymbol{x}_{i}^{[d^{\prime}+1:d+d^{\prime}]}(k)$,
respectively.

\subsection{Dynamic Matrix-valued Edge Weight\label{subsec:Dynamic-Matrix-valued-Edge}}

In order to achieve a tradeoff between privacy-preservation and average
consensus, a periodic switching mechanism for matrix-valued edge weights
is employed, where the period satisfies $T=d+d^{\prime}-1$. Here,
the edge weight matrix will become time-dependent and denoted as $A_{ij}(k)$.
The purpose of this design is to guarantee average consensus of the
multi-agent system \eqref{equ:matrix-consensus-overall} \textcolor{black}{and
the matrix-valued edge weights will be given subsequently.} The necessity
of introducing switching mechanism for matrix-valued edge weight is
that the average consensus can not be guaranteed on a connected time-invariant
matrix-weighted network with positive semi-definite matrix-valued
edge weights, preferable for privacy-preservation purposes. We provide
the following example to illustrate this point.
\begin{example}
\label{exa:cluster-consensus}Consider the matrix-weighted network
$\mathcal{G}$ shown in Figure \ref{fig:integral-network}. Choose
the following positive semi-definite matrix
\begin{equation}
A_{ij}=\left(\begin{array}{ccc}
1 & 1 & 0\\
1 & 1 & 0\\
0 & 0 & 0
\end{array}\right),\label{eq:semi-definite-weight}
\end{equation}
\textcolor{black}{as edge weight for all $(i,j)\in\mathcal{E}$ and
$\sigma=0.1$ that satisfies Lemma}\textcolor{red}{{} }\textcolor{black}{\ref{lem:sigma}.}
Although the network in Figure \ref{fig:integral-network} is connected,
the multi-agent system \eqref{equ:matrix-consensus-overall} only
admits cluster consensus; see Figure \ref{fig:cluster}.
\end{example}
\begin{figure}
\begin{centering}
\includegraphics[width=9cm]{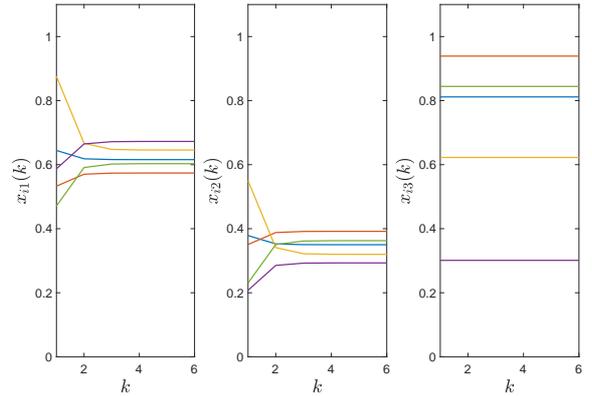}
\par\end{centering}
\caption{Cluster consensus of multi-agent system \eqref{equ:matrix-consensus-overall}
on the connected network in Figure \ref{fig:integral-network} with
the matrix-valued edge weight as in \eqref{eq:semi-definite-weight}.}
\label{fig:cluster}
\end{figure}

According to Lemma \ref{lem:1-1}, the null space of a matrix-valued
Laplacian $L$ is not only determined by the network connectivity,
but also by the null space of weight matrices \citet{pan2018bipartite,trinh2018matrix}.
One can see that $\mathcal{H}$ is determined by the null space of
matrix-valued weights matrices on edges. In this case, the semi-positive
definiteness of static edge weights cannot guarantee that the null
space of a matrix-valued Laplacian is  $\mathcal{R}$, required for
ensuring average consensus under \eqref{equ:matrix-consensus-overall}.
For instance, if one chooses $\boldsymbol{x}(0)\in\mathcal{H}$, then
$\boldsymbol{x}(k)=\boldsymbol{x}(0)$ for any $k\in\mathbb{N}$ and
$\sigma\in\mathbb{R}$, in which case average consensus of multi-agent
\eqref{equ:matrix-consensus-overall} cannot be achieved \citet{pan2018bipartite}.

We shall proceed to present the second ingredient regarding the dynamic
matrix-valued weight design.

Step 2: Dynamic Matrix-valued Edge Weight Design. 

To design dynamic matrix-valued edge weights, we shall employ the
following auxiliary vector set, which is distributed to each agent
at the initialization of the algorithm.

Choose $d^{\prime}\ge3$ and design orthogonal vector set
\begin{equation}
\mathfrak{V}=\left\{ \boldsymbol{v}_{i}\in\mathbb{R}^{d+d^{\prime}}\right\} _{i=1}^{d+d^{\prime}}\label{eq:auxiliary-vector}
\end{equation}
satisfying, 

1) all vectors in $\mathfrak{V}$ are mutually perpendicular,

2) the number of non-zero elements in each $\boldsymbol{v}_{i}\in\mathbb{R}^{d+d^{\prime}}$
is not less than $2$,

3) the number of non-zero entries in $\boldsymbol{v}_{1}$ is less
than $d^{\prime}$,

4) all entries in $\boldsymbol{v}_{d+d^{\prime}}$ are non-zero. 

Based on the vector set $\mathfrak{V}$, we now design the following
mechanism for the periodic switching of matrix-valued edge weights.

For $k=0$, if $(i,j)\in\mathcal{E}$ and $i\in\mathcal{V}_{\mathcal{B}}$
and $j\in\mathcal{N}_{i}\cap\mathcal{V}_{\mathcal{A}}$, choose
\begin{align}
A_{ij}(0) & =A_{ji}(0)\nonumber \\
 & =\alpha_{ij}\frac{\boldsymbol{v}_{1}\boldsymbol{v}_{1}^{\top}}{\boldsymbol{v}_{1}^{\top}\boldsymbol{v}_{1}}+\beta_{ij}\frac{\boldsymbol{v}_{d+d^{\prime}}\boldsymbol{v}_{d+d^{\prime}}^{\top}}{\boldsymbol{v}_{d+d^{\prime}}^{\top}\boldsymbol{v}_{d+d^{\prime}}},\label{eq:edge-weight-matrix-1}
\end{align}
where $\alpha_{ij}>0$ and $\beta_{ij}>0$; else choose $A_{ij}(0)\in\mathbb{R}^{d\times d}$
arbitrarily.

For $k\in\mathbb{Z}_{+}$, if $(i,j)\in\mathcal{E}$, choose
\begin{align}
A_{ij}(k) & =A_{ji}(k)\nonumber \\
 & =\gamma_{ij}^{\rho(k)}\frac{\boldsymbol{v}_{\rho(k)}\boldsymbol{v}_{\rho(k)}^{\top}}{\boldsymbol{v}_{\rho(k)}^{\top}\boldsymbol{v}_{\rho(k)}}+\zeta_{ij}^{\rho(k)}\frac{\boldsymbol{v}_{d+d^{\prime}}\boldsymbol{v}_{d+d^{\prime}}^{\top}}{\boldsymbol{v}_{d+d^{\prime}}^{\top}\boldsymbol{v}_{d+d^{\prime}}},\label{eq:edge-weight-matrix}
\end{align}
where 
\begin{equation}
\rho(k)=\begin{cases}
k\,\text{{\bf mod}}\,d^{*}, & k\,\text{{\bf mod}}\,d^{*}\not=0\\
d^{*}, & k\,\text{{\bf mod}}\,d^{*}=0
\end{cases},\label{eq:gamma1}
\end{equation}

\begin{equation}
\frac{1}{4(n-1)\sigma}>\gamma_{ij}^{\rho(k)}>0,\label{eq:gamma2}
\end{equation}

\begin{equation}
\frac{1}{4(n-1)\sigma}>\zeta_{ij}^{\rho(k)}>0,\label{eq:gamma3}
\end{equation}
and
\[
d^{*}=d+d^{\prime}-1.
\]

The above procedures can be summarized in the Algorithm \ref{MWN-PPAC algorithm}.

\begin{algorithm}
\begin{algorithmic}[1]
\Require{}
\State{Set $k=0$ and choose $d^{\prime}\ge3$.}
\State{Each agent $i\in\mathcal{V}$ constructs $\mathfrak{V}$ in \eqref{eq:auxiliary-vector} and lifts its state.}
\For {$i\in\mathcal{V}$ and $j\in\mathcal{N}_i$}
\If{$ i \in \mathcal{V}_{\mathcal{B}}$ and $j\in \mathcal{N}_i \cap \mathcal{V}_{\mathcal{A}}$}
\State{Construct $A_{ji}(0)$ by \eqref{eq:edge-weight-matrix-1}.}
\Else
\State{Construct arbitrary $A_{ji}(0) \in \mathbb{R}^{d \times d}$.}
\EndIf
\State{Agent $i$ sends $A_{ji}(0)\boldsymbol{x}_{i}(0)$ to $j$.}
\EndFor
\State{Set $k=k+1$}
\Ensure{}
\While{$0<k<\text{max\_iteration\_number}$}
\For {$i\in\mathcal{V}$ and $j\in\mathcal{N}_i$}
\State{Construct $A_{ji}(k)$ by \eqref{eq:edge-weight-matrix}.}
\State{Agent $i$ sends $A_{ji}(k)\boldsymbol{x}_{i}(k)$ to $j$.}
\EndFor
\State{Agent $i$ updates $\boldsymbol{x}_{i}(k)$ using protocol  \eqref{equ:matrix-consensus-protocol}.}
\State{Set $k = k + 1$.}
\EndWhile
\end{algorithmic}\caption{MWN-PPAC algorithm}
\label{MWN-PPAC algorithm}
\end{algorithm}

\begin{rem}[\textcolor{black}{Feasibility of orthogonal vector set $\mathfrak{V}$}]
\label{rem:feasibility of B} Note that the construction of the orthogonal
vector set $\mathfrak{V}$ is always feasible. For convenience, we
provide the following candidate method for constructing $\mathfrak{V}$. 

Choose $d^{\prime}=3$ and one can construct the orthogonal vector
set by choosing
\begin{equation}
\boldsymbol{v}_{i}=(\frac{1}{i},\overset{i-1}{\overbrace{-\frac{1}{i},\ldots,-\frac{1}{i}}},\underset{f_{i}}{\underbrace{1}},\overset{g_{i}}{\overbrace{0,\ldots,0}})^{\top},\label{eq:auxiliary-vector-1}
\end{equation}
for all $i\in\underline{d+3}$ where $g_{i}=\max\left\{ d+2-i,0\right\} $
and $f_{i}=\delta(d+3-i)$. Specifically, consider a matrix-weighted
network $\mathcal{G}$ consisting $n=5$ agents where the state dimension
of each agent is $d=3$. Then, by using \eqref{eq:auxiliary-vector-1},
one can obtain the elements in $\mathfrak{V}=\left\{ \boldsymbol{v}_{i}\in\mathbb{R}^{6}\right\} _{i=1}^{6}$
as follows,
\begin{align}
\boldsymbol{v}_{1} & =\left(1,1,0,0,0,0\right)^{\top},\label{eq:v1}\\
\boldsymbol{v}_{2} & =\left(\frac{1}{2},-\frac{1}{2},1,0,0,0\right)^{\top},\\
\boldsymbol{v}_{3} & =\left(\frac{1}{3},-\frac{1}{3},-\frac{1}{3},1,0,0\right)^{\top},\\
\boldsymbol{v}_{4} & =\left(\frac{1}{4},-\frac{1}{4},-\frac{1}{4},-\frac{1}{4},1,0\right)^{\top},\\
\boldsymbol{v}_{5} & =\left(\frac{1}{5},-\frac{1}{5},-\frac{1}{5},-\frac{1}{5},-\frac{1}{5},1\right)^{\top},
\end{align}
and
\begin{equation}
\boldsymbol{v}_{6}=\left(\frac{1}{6},-\frac{1}{6},-\frac{1}{6},-\frac{1}{6},-\frac{1}{6},-\frac{1}{6}\right)^{\top}.\label{eq:v6}
\end{equation}
\end{rem}
\begin{rem}
According to equations \eqref{eq:edge-weight-matrix-1} to \eqref{eq:gamma1},
the rank of $A_{ij}(k)$ is always $2$ regardless of $d+d^{\prime}$.
Remarkably, the low rank, positive semi-definite matrix-valued edge
weight mechanism is plausible since it renders Algorithm \ref{MWN-PPAC algorithm}
efficient in terms of both communication and computation.
\end{rem}
\begin{rem}
According to Algorithm \ref{MWN-PPAC algorithm}, the observations
of an agent $i\in\mathcal{V}$ satisfies
\begin{equation}
\mathcal{O}_{i}^{\mathcal{M}}(\{\boldsymbol{x}_{l}(0),\mathcal{I}_{l}\}_{l\in\mathcal{V}}))=\left\{ \text{form of Algorithm 1}\right\} \cup\mathcal{P}\cup\mathcal{S}^{i},\label{eq:observation-of-VA}
\end{equation}
where
\begin{equation}
\mathcal{P}=\{\sigma,\mathfrak{V},d,d^{\prime}\}\label{eq:public-to-all}
\end{equation}
and
\begin{equation}
\mathcal{S}^{i}=\{\text{\ensuremath{\boldsymbol{x}_{i}}}(k),\mathcal{N}_{i},A_{ij}(k),A_{ji}(k),\boldsymbol{y}_{j\rightarrow i}(k)\}_{j\in\mathcal{N}_{i},k\in\mathbb{N}}.\label{eq:Si}
\end{equation}
\end{rem}
In the following, we shall first analyze the asymptotic convergence
of the Algorithm \ref{MWN-PPAC algorithm} towards average consensus.

\section{Average Consensus Analysis \label{sec:Consensus}}

In this section, we shall show the asymptotic average consensus of
matrix-weighted network \eqref{equ:matrix-consensus-overall} can
be achieved under Algorithm \ref{MWN-PPAC algorithm}. Since the matrix-valued
edge weights are dynamically switching, the analysis is equivalent
to the average consensus problem of multi-agent system \eqref{equ:matrix-consensus-overall}
on matrix-weighted switching networks; the dynamic edge weights constructed
in Algorithm \ref{MWN-PPAC algorithm} are just a special case of
our convergence analysis for average consensus. To keep the following
presentation concise, we shall occasionally write $A_{k}$ and $L_{k}$
for matrix-valued weight matrix $A(k)$ and matrix-valued Laplacian
$L(k)$, respectively.

The main result in this section is stated in the following theorem.
\begin{thm}
\label{thm:consensus theorem} The matrix-weighted network \eqref{equ:matrix-consensus-overall}
using Algorithm 1 achieves asymptotic average consensus, namely, ${\color{blue}{\color{black}{\displaystyle \lim_{k\rightarrow\infty}}\boldsymbol{x}(k)=\mathds{1}_{n}\otimes\left(\text{Avg}(\boldsymbol{x}(0))\right)}}.$
\end{thm}
In order to prove Theorem \ref{thm:consensus theorem}, we need to
establish the following auxiliary results.

First, we note that the structure of the null space of a matrix-valued
Laplacian matrix for the time-invariant network essentially determines
the steady-state of the multi-agent system \eqref{equ:matrix-consensus-overall};
refer to equations \eqref{eq:consensus-space} and \eqref{eq:edge-space}.
According to Lemma \ref{lem:1-1} and Lemma \ref{lem:1-2}, it is
reasonable to examine the null space of the matrix-weighted Laplacian
corresponding to a series of matrix-weighted networks. For $k^{\prime},k^{\prime\prime}\in\mathbb{Z}_{+}$
and $k^{\prime}<k^{\prime\prime}$, the union of graphs $\mathcal{G}(k)=(\mathcal{V},\mathcal{E},A_{k})$
over a time interval $[k^{\prime},k^{\prime\prime}]\subset[0,\infty)$
is, 
\[
\bigcup_{k=k^{\prime}}^{k^{\prime\prime}-1}\mathcal{G}(k)=\left(\mathcal{V},\mathcal{E},\sum_{k=k^{\prime}}^{k^{\prime\prime}-1}A_{k}\right).
\]

The following result reveals the connection between the null space
of matrix-valued Laplacian corresponding to the union of a series
of matrix-weighted networks and the intersection of the null space
of matrix-valued Laplacians corresponding to each separate matrix-weighted
network.
\begin{thm}
\label{thm:nullspace-relationship}Let $\mathcal{G}_{k}=(\mathcal{V},\mathcal{E},A_{k})$
be a matrix-weighted switching network with matrix-weighted Laplacian
$L_{k}$. Then 
\[
\text{{\bf null}}\left(\sum_{i=1}^{k^{\prime\prime}-k^{\prime}}L_{k^{\prime}+i-1}\right)=\mathcal{R}
\]
if and only if 
\[
\underset{i\in\underline{k^{\prime\prime}-k^{\prime}}}{\bigcap}\text{{\bf null}}\left(L_{k^{\prime}+i-1}\right)=\mathcal{R},
\]
where $k^{\prime}<k^{\prime\prime}\in\mathbb{Z}_{+}$.
\end{thm}
\begin{proof}
(Necessity) From the definition of matrix-valued Laplacian matrix,
one has 
\[
\mathcal{R}\subseteq\underset{i\in\underline{k^{\prime\prime}-k^{\prime}}}{\bigcap}\text{{\bf null}}\left(L_{k^{\prime}+i-1}\right).
\]
Assume that $\underset{i\in\underline{k^{\prime\prime}-k^{\prime}}}{\bigcap}\text{{\bf null}}\left(L_{k^{\prime}+i-1}\right)\neq\mathcal{R}$;
then there exists an $\boldsymbol{\eta}\notin\mathcal{R}$ such that
$L_{k^{\prime}+i-1}\boldsymbol{\eta}=\boldsymbol{0}$ for all $i\in\underline{k^{\prime\prime}-k^{\prime}}$,
which would imply,
\begin{align*}
\left(\sum_{i=1}^{k^{\prime\prime}-k^{\prime}}L_{k^{\prime}+i-1}\right)\boldsymbol{\eta} & =\sum_{i=1}^{k^{\prime\prime}-k^{\prime}}\left(L_{k^{\prime}+i-1}\boldsymbol{\eta}\right)\\
 & =\boldsymbol{0},
\end{align*}
contradicting the fact that 
\[
\text{{\bf null}}\left(\sum_{i=1}^{k^{\prime\prime}-k^{\prime}}L_{k^{\prime}+i-1}\right)=\mathcal{R}.
\]
Therefore, 
\[
\underset{i\in\underline{k^{\prime\prime}-k^{\prime}}}{\bigcap}\text{{\bf null}}\left(L_{k^{\prime}+i-1}\right)=\mathcal{R}.
\]

(Sufficiency) Assume that 
\[
\text{{\bf null}}\left(\sum_{i=1}^{k^{\prime\prime}-k^{\prime}}L_{k^{\prime}+i-1}\right)\neq\mathcal{R}.
\]
Then there exists $\boldsymbol{\eta}\notin\mathcal{R}$ such that
\[
\left(\sum_{i=1}^{k^{\prime\prime}-k^{\prime}}L_{k^{\prime}+i-1}\right)\boldsymbol{\eta}=\boldsymbol{0}.
\]
Hence, 
\[
\boldsymbol{\eta}^{\top}\left(\sum_{i=1}^{k^{\prime\prime}-k^{\prime}}L_{k^{\prime}+i-1}\right)\boldsymbol{\eta}=0
\]
implying that
\begin{align*}
 & \boldsymbol{\eta}^{\top}\left(\sum_{i=1}^{k^{\prime\prime}-k^{\prime}}L_{k^{\prime}+i-1}\right)\boldsymbol{\eta}=\sum_{i=1}^{k^{\prime\prime}-k^{\prime}}\boldsymbol{\eta}^{\top}L_{k^{\prime}+i-1}\boldsymbol{\eta}=0.
\end{align*}
We note that $L_{k^{\prime}+i-1}$ is positive semi-definite for all
$i\in\underline{k^{\prime\prime}-k^{\prime}}$, and $\boldsymbol{\eta}^{\top}L_{k^{\prime}+i-1}\boldsymbol{\eta}=0$,
one has $L_{k^{\prime}+i-1}\boldsymbol{\eta}=\boldsymbol{0}$; this,
on the other hand, contradicts the premise that 
\[
\underset{i\in\underline{k^{\prime\prime}-k^{\prime}}}{\bigcap}\text{{\bf null}}\left(L_{k^{\prime}+i-1}\right)=\mathcal{R}.
\]
Thus 
\[
\text{{\bf null}}\left(\sum_{i=1}^{k^{\prime\prime}-k^{\prime}}L_{k^{\prime}+i-1}\right)=\mathcal{R}.
\]
\end{proof}
\begin{rem}
Remarkably, Theorem \ref{thm:nullspace-relationship} establishes
a quantitative connection between the null space of matrix-valued
Laplacian of a set of matrix-weighted networks and that of their union.
This is eventually a result that is valid for general switching matrix-weighted
networks whose edge weight matrices can be either positive definite
or positive semi-definite. Essentially, the null space of matrix-valued
Laplacian of a set of matrix-weighted networks is equal to that of
their union, i.e., $\text{{\bf null}}\left(\sum_{i=1}^{k^{\prime\prime}-k^{\prime}}L_{k^{\prime}+i-1}\right)=\underset{i\in\underline{k^{\prime\prime}-k^{\prime}}}{\bigcap}\text{{\bf null}}\left(L_{k^{\prime}+i-1}\right)$.
\end{rem}
We now proceed to prove the asymptotic average consensus of the matrix-weighted
network \eqref{equ:matrix-consensus-overall} under Algorithm 1. The
idea is to construct the error vector 
\begin{equation}
\boldsymbol{\omega}(k)=\boldsymbol{x}(k)-\mathds{1}_{n}\otimes\left(\text{Avg}(\boldsymbol{x}(0))\right),\label{eq:error-vector}
\end{equation}
and then show that $\boldsymbol{\omega}(k)$ converges to the origin
as $k$ goes to infinity.

To this end, we first need to employ the state transition matrix from
time $k^{\prime}$ to time $k^{\prime\prime}$, denoted by 
\[
\varPhi(k^{\prime\prime},k^{\prime})=\prod_{i=1}^{k^{\prime\prime}-k^{\prime}}\left(I-\sigma L_{k^{\prime}+i-1}\right).
\]
Then $\boldsymbol{x}(k^{\prime\prime})=\varPhi(k^{\prime\prime},k^{\prime})\boldsymbol{x}(k^{\prime})$,
where $k^{\prime}<k^{\prime\prime}\in\mathbb{Z}_{+}$. 

Note that the matrix-valued Laplacian matrix $L_{k}$ has at least
$d$ zero eigenvalues. Let $\lambda_{1}\leq\lambda_{2}\leq\cdots\leq\lambda_{dn}$
be the eigenvalues of $L_{k}$. Then we have
\[
0=\lambda_{1}=\cdots=\lambda_{d}\leq\lambda_{d+1}\le\cdots\leq\lambda_{dn}.
\]
Let $\beta_{1}\geq\beta_{2}\geq\cdots\geq\beta_{dn}$ denote the eigenvalues
of $I-\sigma L_{k}$; then $\beta_{i}(I-\sigma L_{k})=1-\sigma\lambda_{i}(L_{k})$,
namely,
\[
1=\beta_{1}=\cdots=\beta_{d}\geq\beta_{d+1}\geq\cdots\geq\beta_{dn}.
\]
In the meantime, the eigenvector corresponding to the eigenvalue $\beta_{i}(I-\sigma L_{k})$
is equal to that associated with $\lambda_{i}(L_{k})$. 

Note that $\varPhi(k^{\prime\prime},k^{\prime})^{\top}\varPhi(k^{\prime\prime},k^{\prime})$
has at least $d$ eigenvalues at $1$. Let $\mu_{j}$ be the eigenvalues
of $\varPhi(t_{k^{\prime\prime}},t_{k^{\prime}})^{\top}\varPhi(t_{k^{\prime\prime}},t_{k^{\prime}})$,
where $j\in\underline{dn}$ such that $\mu_{1}=\cdots=\mu_{d}=1$
and $\mu_{d+1}\geq\mu_{d+2}\geq\cdots\geq\mu_{dn}$. 

We now present the following lemma to establish the connection between
the null space of the matrix 
\[
\sum_{i=1}^{k^{\prime\prime}-k^{\prime}}L_{k^{\prime}+i-1},
\]
and the specific eigenvalue $\mu_{d+1}$ of $\varPhi(k^{\prime\prime},k^{\prime})^{\top}\varPhi(k^{\prime\prime},k^{\prime})$
which turns out to be crucial for determining the convergence of error
vector $\boldsymbol{\omega}(k)$ in \eqref{eq:error-vector}. 
\begin{thm}
\label{thm:eigenvalue inequality}Let $\mathcal{G}_{k}=(\mathcal{V},\mathcal{E},A_{k})$
be a matrix-weighted switching network with matrix-weighted Laplacian
$L_{k}$. Then 
\[
\text{{\bf null}}\left(\sum_{i=1}^{k^{\prime\prime}-k^{\prime}}L_{k^{\prime}+i-1}\right)=\mathcal{R}
\]
if and only if 
\[
\mu_{d+1}\left(\varPhi(k^{\prime\prime},k^{\prime})^{\top}\varPhi(k^{\prime\prime},k^{\prime})\right)<1,
\]
where $k^{\prime}<k^{\prime\prime}\in\mathbb{Z}_{+}$.
\end{thm}
\begin{proof}
(Sufficiency) Assume that $\text{{\bf null}}\left(\sum_{i=1}^{k^{\prime\prime}-k^{\prime}}L_{k^{\prime}+i-1}\right)\neq\mathcal{R}$;
then according to Rayleigh theorem (\citeauthor[Theorem 4.2.2,  p.235]{horn2012matrix}),
there exists an $\boldsymbol{\eta}\notin\mathcal{R}$ such that $L_{k^{\prime}+i-1}\boldsymbol{\eta}=\boldsymbol{0}$
for all $i\in\underline{k^{\prime\prime}-k^{\prime}}$. Thus one can
obtain $\left(I-\sigma L_{k^{\prime}+i-1}\right)\boldsymbol{\eta}=\boldsymbol{\eta}$
for all $i\in\underline{k^{\prime\prime}-k^{\prime}}$ and $\varPhi(k^{\prime\prime},k^{\prime})\boldsymbol{\eta}=\boldsymbol{\eta}$.
Using Rayleigh theorem again, one has
\begin{eqnarray*}
 &  & \mu_{d+1}\left(\varPhi(k^{\prime\prime},k^{\prime})^{\top}\varPhi(k^{\prime\prime},k^{\prime})\right)\\
 & \geq & \frac{\boldsymbol{\eta}^{\top}\varPhi(k^{\prime\prime},k^{\prime})^{\top}\varPhi(k^{\prime\prime},k^{\prime})\boldsymbol{\eta}}{\boldsymbol{\eta}^{\top}\boldsymbol{\eta}}=1,
\end{eqnarray*}
contradicting,
\[
\mu_{d+1}\left(\varPhi(k^{\prime\prime},k^{\prime})^{\top}\varPhi(k^{\prime\prime},k^{\prime})\right)<1.
\]
Therefore $\text{{\bf null}}\left(\sum_{i=1}^{k^{\prime\prime}-k^{\prime}}L_{k^{\prime}+i-1}\right)=\mathcal{R}$
holds.

(Necessity) Assume that $\mu_{d+1}\left(\varPhi(k^{\prime\prime},k^{\prime})^{\top}\varPhi(k^{\prime\prime},k^{\prime})\right)\geq1$.
Again, by Rayleigh theorem, there exists a $\boldsymbol{\eta}\notin\mathcal{R}$
and $\boldsymbol{\eta}\neq\boldsymbol{0}$ such that
\begin{align*}
 & \mu_{d+1}\left(\varPhi(k^{\prime\prime},k^{\prime})^{\top}\varPhi(k^{\prime\prime},k^{\prime})\right)\\
 & =\frac{\boldsymbol{\eta}^{\top}\varPhi(k^{\prime\prime},k^{\prime})^{\top}\varPhi(k^{\prime\prime},k^{\prime})\boldsymbol{\eta}}{\boldsymbol{\eta}^{\top}\boldsymbol{\eta}}\geq1.
\end{align*}
Thus
\[
\parallel\boldsymbol{\eta}\parallel\leq\parallel\varPhi(k^{\prime\prime},k^{\prime})\boldsymbol{\eta}\parallel.
\]

Let $\boldsymbol{\eta}_{k^{\prime}}=\boldsymbol{\eta}$ and $\boldsymbol{\eta}_{k^{\prime}+i}=\left(I-\sigma L_{k^{\prime}+i-1}\right)\boldsymbol{\eta}_{k^{\prime}+i-1}$
for $i\in\underline{k^{\prime\prime}-k^{\prime}}$. Due to the fact
$\beta_{j}\left(I-\sigma L_{k^{\prime}+i-1}\right)\leq1$ \textcolor{black}{and
$\lambda_{j}(L_{k^{\prime}+i-1})<\frac{1}{\sigma}$ (according to
Lemma \ref{lem:sigma} in $\mathsection$\ref{sec:Consensus}}) for
$j\in\underline{dn}$ and $i\in\underline{k^{\prime\prime}-k^{\prime}}$,
then 
\[
\parallel\left(I-\sigma L_{k^{\prime}+i-1}\right)\boldsymbol{\eta}_{k^{\prime}+i-1}\parallel\leq\parallel\boldsymbol{\eta}_{k^{\prime}+i-1}\parallel,
\]
which implies that,
\begin{eqnarray*}
\parallel\boldsymbol{\eta}\parallel & \leq & \parallel\varPhi(k^{\prime\prime},k^{\prime})\boldsymbol{\eta}\parallel\\
 & = & \parallel\boldsymbol{\eta}_{k^{\prime\prime}}\parallel\leq\parallel\boldsymbol{\eta}_{k^{\prime\prime}-1}\parallel\leq\ldots\leq\parallel\boldsymbol{\eta}_{k^{\prime}}\parallel\\
 & = & \parallel\boldsymbol{\eta}\parallel.
\end{eqnarray*}
Hence, 
\[
\parallel\left(I-\sigma L_{k^{\prime}+i-1}\right)\boldsymbol{\eta}_{k^{\prime}+i-1}\parallel=\parallel\boldsymbol{\eta}_{k^{\prime}+i-1}\parallel,
\]
 for any $i\in\underline{k^{\prime\prime}-k^{\prime}}$. Then 
\begin{eqnarray*}
 & \boldsymbol{\eta}_{k^{\prime}+i-1}^{\top} & \left(I-\sigma L_{k^{\prime}+i-1}\right)\left(I-\sigma L_{k^{\prime}+i-1}\right)\boldsymbol{\eta}_{k^{\prime}+i-1}\\
= & \boldsymbol{\eta}_{k^{\prime}+i-1}^{\top} & \boldsymbol{\eta}_{k^{\prime}+i-1}.
\end{eqnarray*}
By Rayleigh theorem, 
\[
\left(I-\sigma L_{k^{\prime}+i-1}\right)\boldsymbol{\eta}_{k^{\prime}+i-1}=\boldsymbol{\eta}_{k^{\prime}+i-1},
\]
and thereby $L_{\rho(k^{\prime}+i-1)}\boldsymbol{\eta}_{k^{\prime}+i-1}=\boldsymbol{0}.$
Thus 
\[
\boldsymbol{\eta}_{k^{\prime}+i-1}\in\text{{\bf null}}(L_{k^{\prime}+i-1}).
\]
Since,
\begin{eqnarray*}
\parallel\boldsymbol{\eta}_{k^{\prime}+i} & - & \boldsymbol{\eta}_{k^{\prime}+i-1}\parallel\\
 & = & \parallel\left(I-\sigma L_{k^{\prime}+i-1}\right)\boldsymbol{\eta}_{k^{\prime}+i-1}-\boldsymbol{\eta}_{k^{\prime}+i-1}\parallel\\
 & = & \parallel\sigma L_{k^{\prime}+i-1}\boldsymbol{\eta}_{k^{\prime}+i-1}\parallel\\
 & = & 0,
\end{eqnarray*}
one can further obtain $\boldsymbol{\eta}_{k^{\prime}+i-1}=\boldsymbol{\eta}_{k^{\prime}+i}$
for any $i\in\underline{k^{\prime\prime}-k^{\prime}}$, which implies
that 
\[
\boldsymbol{\eta}\in\underset{i\in\underline{k^{\prime\prime}-k^{\prime}}}{\bigcap}\text{{\bf null}}(L_{k^{\prime}+i-1})
\]
 and 
\[
\text{{\bf null}}\left(\sum_{i=1}^{k^{\prime\prime}-k^{\prime}}L_{k^{\prime}+i-1}\right)\neq\mathcal{R},
\]
which is a contradiction. The proof is thus concluded.
\end{proof}
\begin{rem}
Similarly, Theorem \ref{thm:eigenvalue inequality} is also a rather
general result that is valid for switching matrix-weighted networks
whose edge weight matrices can be either positive definite or positive
semi-definite.
\end{rem}
We are now ready to examine the connection between achieving asymptotic
average consensus and the dynamic matrix-valued edge weights constructed
by Algorithm 1. In Algorithm 1, $\rho(k)$ is employed to denote the
periodic switching signal for matrix-valued edge weights. Then, the
matrix-valued weight matrix at time $k>0$ satisfies $A(k)=A(\rho(k))$
and the matrix-valued Laplacian also has $L(k)=L(\rho(k))$. The set
of candidate matrix-valued edge weights is hence,
\begin{equation}
\mathbb{A}_{\text{c}}=\left\{ A(1),A(2),\cdots,A(d+d^{\prime}-1)\right\} .
\end{equation}
We proceed to present the following lemma related to the null space
of $L_{\rho(k)}$ and the existence of a positive spanning tree in
a periodic switching matrix-weighted network constructed by Algorithm
1.
\begin{lem}
\label{lem: positivity}Let $L_{\rho(k)}$ be the matrix-valued Laplacian
corresponding to matrix-valued edge weights constructed in Algorithm
1. Then 
\begin{equation}
\text{{\bf null}}\left(\sum_{k=1}^{d+d^{\prime}-1}L_{\rho(k)}\right)=\mathcal{R}.
\end{equation}
\end{lem}
\begin{proof}
According to the construction $\left\{ A_{ij}(k)|k>0\right\} $, one
can see that: for $k>0$, $\boldsymbol{v}_{\rho(k)}\in\mathfrak{V}$
is an eigenvector corresponding to the eigenvalue $\gamma_{ij}^{\rho(k)}$
of matrix $A_{ij}(k)$ and $\boldsymbol{v}_{d+d^{\prime}}\in\mathfrak{V}$
is an eigenvector corresponding to the eigenvalue $\zeta_{ij}^{\rho(k)}$
of matrix $A_{ij}(k)$ and elements in $\mathfrak{V}\setminus\left\{ \boldsymbol{v}_{\rho(k)},\,\boldsymbol{v}_{d+d^{\prime}}\right\} $
are eigenvectors corresponding to zero eigenvalue of matrix $A_{ij}(k)$,
i.e., $\text{{\bf rank}}\left(A_{ij}(k)\right)=2$ for $k>0$. Consequently,
the null spaces of $A_{ij}(k)$'s lead to having,
\begin{equation}
\bigcap_{k=1}^{d+d^{\prime}-1}\text{{\bf null}}(A_{ij}(k))=\left\{ \boldsymbol{0}\right\} ,
\end{equation}
for all $(i,j)\in\mathcal{E}$. Therefore, $\sum_{k=1}^{d+d^{\prime}-1}A_{ij}(k)$
is positive definite for all $(i,j)\in\mathcal{E}$, i.e., the union
of $\mathcal{G}(k)$ over the time interval $\left[1,d+d^{\prime}-1\right]$
has a positive spanning tree, according to Lemma \ref{lem:1-2}, $\text{{\bf null}}\left(\sum_{k=1}^{d+d^{\prime}-1}L_{\rho(k)}\right)=\mathcal{R}$,
thus completing the proof.
\end{proof}
\begin{lem}
\label{lem:sigma}Let $L_{\rho(k)}$ be the matrix-valued Laplacian
corresponding to matrix-valued edge weights constructed in Algorithm
1, where $k\in\mathbb{Z}_{+}$. If
\begin{equation}
\frac{1}{4(n-1)\sigma}>\gamma_{ij}^{\rho(k)}>0,\label{eq:sigma-2-1}
\end{equation}
and 
\begin{equation}
\frac{1}{4(n-1)\sigma}>\zeta_{ij}^{\rho(k)}>0,\label{eq:sigma-2-2}
\end{equation}
for any $(i,j)\in\mathcal{E}$, then $\lambda_{\text{max}}\left(L_{\rho(k)}\right)<\frac{1}{\sigma}$.
\end{lem}
\begin{proof}
Let
\[
\Gamma^{\rho(k)}=\left[\begin{array}{cccc}
{\displaystyle \sum_{j=1}^{n}}\gamma_{1j}^{\rho(k)}, & -\gamma_{12}^{\rho(k)}, & \ldots, & -\gamma_{1n}^{\rho(k)}\\
-\gamma_{21}^{\rho(k)}, & {\displaystyle \sum_{j=1}^{n}}\gamma_{2j}^{\rho(k)}, & \ldots, & -\gamma_{2n}^{\rho(k)}\\
\vdots & \vdots & \ddots & \vdots\\
-\gamma_{n1}^{\rho(k)}, & -\gamma_{n2}^{\rho(k)}, & \ldots & {\displaystyle \sum_{j=1}^{n}}\gamma_{nj}^{\rho(k)}
\end{array}\right],
\]
and 
\[
\Pi^{\rho(k)}=\left[\begin{array}{cccc}
{\displaystyle \sum_{j=1}^{n}}\zeta_{1j}^{\rho(k)}, & -\zeta_{12}^{\rho(k)}, & \ldots, & -\zeta_{1n}^{\rho(k)}\\
-\zeta_{21}^{\rho(k)}, & {\displaystyle \sum_{j=1}^{n}}\zeta_{2j}^{\rho(k)}, & \ldots, & -\zeta_{2n}^{\rho(k)}\\
\vdots & \vdots & \ddots & \vdots\\
-\zeta_{n1}^{\rho(k)}, & -\zeta_{n2}^{\rho(k)}, & \ldots & {\displaystyle \sum_{j=1}^{n}}\zeta_{nj}^{\rho(k)}
\end{array}\right].
\]
Note that matrix-valued Laplacian corresponding to matrix-valued edge
weights constructed in Algorithm 1 satisfy
\begin{eqnarray*}
L_{\rho(k)} & = & \Gamma^{\rho(k)}\otimes\frac{\boldsymbol{v}_{\rho(k)}\boldsymbol{v}_{\rho(k)}^{\top}}{\boldsymbol{v}_{\rho(k)}^{\top}\boldsymbol{v}_{\rho(k)}}\\
 & + & \Pi^{\rho(k)}\otimes\frac{\boldsymbol{v}_{d+d^{\prime}}\boldsymbol{v}_{d+d^{\prime}}^{\top}}{\boldsymbol{v}_{d+d^{\prime}}^{\top}\boldsymbol{v}_{d+d^{\prime}}};
\end{eqnarray*}
therefore, 

\begin{eqnarray*}
 &  & \lambda_{\text{max}}\left(L_{\rho(k)}\right)\\
 & = & \parallel L_{\rho(k)}\parallel_{2}\\
 & \leq & \left\Vert \Gamma^{\rho(k)}\otimes\frac{\boldsymbol{v}_{\rho(k)}\boldsymbol{v}_{\rho(k)}^{\top}}{\boldsymbol{v}_{\rho(k)}^{\top}\boldsymbol{v}_{\rho(k)}}\right\Vert _{2}\\
 & + & \left\Vert \Pi^{\rho(k)}\otimes\frac{\boldsymbol{v}_{d+d^{\prime}}\boldsymbol{v}_{d+d^{\prime}}^{\top}}{\boldsymbol{v}_{d+d^{\prime}}^{\top}\boldsymbol{v}_{d+d^{\prime}}}\right\Vert _{2}\\
 & = & \lambda_{\text{max}}\left(\Gamma^{\rho(k)}\otimes\frac{\boldsymbol{v}_{\rho(k)}\boldsymbol{v}_{\rho(k)}^{\top}}{\boldsymbol{v}_{\rho(k)}^{\top}\boldsymbol{v}_{\rho(k)}}\right)\\
 & + & \lambda_{\text{max}}\left(\Pi^{\rho(k)}\otimes\frac{\boldsymbol{v}_{d+d^{\prime}}\boldsymbol{v}_{d+d^{\prime}}^{\top}}{\boldsymbol{v}_{d+d^{\prime}}^{\top}\boldsymbol{v}_{d+d^{\prime}}}\right)\\
 & = & \lambda_{\text{max}}\left(\Gamma^{\rho(k)}\right)\lambda_{\text{max}}\left(\frac{\boldsymbol{v}_{\rho(k)}\boldsymbol{v}_{\rho(k)}^{\top}}{\boldsymbol{v}_{\rho(k)}^{\top}\boldsymbol{v}_{\rho(k)}}\right)\\
 & + & \lambda_{\text{max}}\left(\Pi^{\rho(k)}\right)\lambda_{\text{max}}\left(\frac{\boldsymbol{v}_{d+d^{\prime}}\boldsymbol{v}_{d+d^{\prime}}^{\top}}{\boldsymbol{v}_{d+d^{\prime}}^{\top}\boldsymbol{v}_{d+d^{\prime}}}\right).
\end{eqnarray*}
Note that,
\begin{align*}
\lambda_{\text{max}}\left(\frac{\boldsymbol{v}_{\rho(k)}\boldsymbol{v}_{\rho(k)}^{\top}}{\boldsymbol{v}_{\rho(k)}^{\top}\boldsymbol{v}_{\rho(k)}}\right) & =\lambda_{\text{max}}\left(\frac{\boldsymbol{v}_{d+d^{\prime}}\boldsymbol{v}_{d+d^{\prime}}^{\top}}{\boldsymbol{v}_{d+d^{\prime}}^{\top}\boldsymbol{v}_{d+d^{\prime}}}\right)\\
 & =1.
\end{align*}
Then one has,
\[
\lambda_{\text{max}}\left(L_{\rho(k)}\right)\leq\lambda_{\text{max}}\left(\Gamma^{\rho(k)}\right)+\lambda_{\text{max}}\left(\Pi^{\rho(k)}\right).
\]
According to \eqref{eq:sigma-2-1} and \eqref{eq:sigma-2-2}, one
has,
\begin{align*}
\left\Vert \Gamma^{\rho(k)}\right\Vert _{1} & =\left\Vert \Gamma^{\rho(k)}\right\Vert _{\infty}\\
 & =\max_{1<j<n}\left(2{\displaystyle \sum_{i=1}^{n}}\gamma_{ij}^{\rho(k)}\right)<\frac{1}{2\sigma},
\end{align*}
and

\begin{align*}
\left\Vert \Pi^{\rho(k)}\right\Vert _{1} & =\left\Vert \Pi^{\rho(k)}\right\Vert _{\infty}\\
 & =\max_{1<j<n}\left(2{\displaystyle \sum_{i=1}^{n}}\zeta_{ij}^{\rho(k)}\right)<\frac{1}{2\sigma}.
\end{align*}
Using the facts that
\begin{align*}
\lambda_{\text{max}}\left(\Gamma^{\rho(k)}\right) & =\left\Vert \Gamma^{\rho(k)}\right\Vert _{2}\\
 & \leq\sqrt{\left\Vert \Gamma^{\rho(k)}\right\Vert _{1}\left\Vert \Gamma^{\rho(k)}\right\Vert _{\infty}},
\end{align*}
and
\begin{align*}
\lambda_{\text{max}}\left(\Pi^{\rho(k)}\right) & =\left\Vert \Pi^{\rho(k)}\right\Vert _{2}\\
 & \leq\sqrt{\left\Vert \Pi^{\rho(k)}\right\Vert _{1}\left\Vert \Pi^{\rho(k)}\right\Vert _{\infty}}
\end{align*}
yields $\lambda_{\text{max}}\left(\Gamma^{\rho(k)}\right)<\frac{1}{2\sigma}$
and $\lambda_{\text{max}}\left(\Pi^{\rho(k)}\right)<\frac{1}{2\sigma}$.
Therefore, 
\[
\lambda_{\text{max}}\left(L_{\rho(k)}\right)<\frac{1}{\sigma}.
\]
\end{proof}
Note that $\rho(k)$ in Lemma \ref{lem:sigma} is eventually a periodic
function of time index $k$. Therefore, $L_{\rho(k)}$ in Lemma \ref{lem:sigma}
is related to the matrix-value Laplacian at time $k$ and can also
be referred to as $L(k)$.

We are now ready to prove Theorem \ref{thm:consensus theorem}.
\begin{proof}
Let 
\[
\boldsymbol{\omega}(k)=\boldsymbol{x}(k)-\boldsymbol{x}_{f},
\]
where $\boldsymbol{x}_{f}=\mathds{1}_{n}\otimes\left(\text{Avg}(\boldsymbol{x}(0))\right)$.
Then we have 
\[
\boldsymbol{\omega}(k+1)=\left(I-\sigma L(k)\right)\boldsymbol{\omega}(k).
\]
Denote 
\[
\varPhi(d+d^{\prime},1)=\left(I-\sigma L(d+d^{\prime}-1)\right)\cdots\left(I-\sigma L(1)\right),
\]
If $\boldsymbol{\omega}(1)\in\mathcal{R}$, then $\boldsymbol{x}(k)=\boldsymbol{x}_{f}$
for any $k\geq1$; else if $\boldsymbol{\omega}(1)\notin\mathcal{R}$,
observe that,
\begin{align*}
 & \mu_{d+d^{\prime}+1}(\varPhi(d+d^{\prime},1)^{\top}\varPhi(d+d^{\prime},1))\\
 & \geq\frac{\boldsymbol{\omega}(1)^{\top}(\varPhi(d+d^{\prime},1)^{\top}\varPhi(d+d^{\prime},1))\boldsymbol{\omega}(1)}{\boldsymbol{\omega}(1)^{\top}\boldsymbol{\omega}(1)}\\
 & =\frac{\boldsymbol{\omega}(d+d^{\prime})^{\top}\boldsymbol{\omega}(d+d^{\prime})}{\boldsymbol{\omega}(1)^{\top}\boldsymbol{\omega}(1)},
\end{align*}
which implies,
\begin{align*}
 & \boldsymbol{\omega}(d+d^{\prime})^{\top}\boldsymbol{\omega}(d+d^{\prime})\\
 & \leq\mu_{d+d^{\prime}+1}(\varPhi(d+d^{\prime},1)^{\top}\varPhi(d+d^{\prime},1))\boldsymbol{\omega}(1)^{\top}\boldsymbol{\omega}(1).
\end{align*}
Therefore,
\begin{align*}
 & \parallel\boldsymbol{\omega}(d+d^{\prime})\parallel\\
 & \leq\mu_{d+d^{\prime}+1}(\varPhi(d+d^{\prime},1)^{\top}\varPhi(d+d^{\prime},1))^{\frac{1}{2}}\parallel\boldsymbol{\omega}(1)\parallel,
\end{align*}
if there exists $k_{0}\in\mathbb{N}$ such that $\boldsymbol{\omega}({\color{red}{\color{black}(k_{0}+1)}}(d+d^{\prime}))\in\mathcal{R}$,
then $\boldsymbol{x}(k)=\boldsymbol{x}_{f}$ for any $k\geq{\color{red}{\color{black}(k_{0}+1)}}(d+d^{\prime})$.
Otherwise, for any $p\in\mathbb{Z}_{+}$, one has,
\begin{align*}
 & \parallel\boldsymbol{\omega}(p(d+d^{\prime}))\parallel\\
 & \leq\mu_{d+d^{\prime}+1}(\varPhi(d+d^{\prime},1)^{\top}\varPhi(d+d^{\prime},1))^{\frac{1}{2}p}\parallel\boldsymbol{\omega}(1)\parallel.
\end{align*}
Moreover,
\begin{align*}
 & \parallel\boldsymbol{\omega}(k+1)\parallel-\parallel\boldsymbol{\omega}(k)\parallel\\
= & \boldsymbol{\omega}(k)^{\top}\left(I-\sigma L(k)\right)\left(I-\sigma L(k)\right)\boldsymbol{\omega}(k)-\boldsymbol{\omega}(k)^{\top}\boldsymbol{\omega}(k)\\
= & -\boldsymbol{\omega}(k)^{\top}\left(2\sigma L(k)-\sigma^{2}L^{2}(k)\right)\boldsymbol{\omega}(k),
\end{align*}
and 
\[
\frac{\boldsymbol{\omega}(k)^{\top}L^{2}(k)\boldsymbol{\omega}(k)}{\boldsymbol{\omega}(k)^{\top}L(k)\boldsymbol{\omega}(k)}\leq\lambda_{\text{max}}(L(k)).
\]
From Lemma \ref{lem:sigma}, $\lambda_{\text{max}}(L(k))<\frac{1}{\sigma}$,
thus, 
\[
\sigma^{2}\boldsymbol{\omega}(k)^{\top}L^{2}(k)\boldsymbol{\omega}(k)<\sigma\boldsymbol{\omega}(k)^{\top}L(k)\boldsymbol{\omega}(k),
\]
and
\[
\parallel\boldsymbol{\omega}(k+1)\parallel-\parallel\boldsymbol{\omega}(k)\parallel<0.
\]
Hence,
\begin{align*}
 & \parallel\boldsymbol{\omega}(k)\parallel\\
\leq & \parallel\boldsymbol{\omega}(p(d+d^{\prime}))\parallel\\
\leq & \mu_{d+d^{\prime}+1}(\varPhi(d+d^{\prime},1)^{\top}\varPhi(d+d^{\prime},1))^{\frac{1}{2}p}\parallel\boldsymbol{\omega}(1)\parallel,
\end{align*}
for $k\in[p(d+d^{\prime}),(p+1)(d+d^{\prime}))$ and any $p\in\mathbb{Z}_{+}$.

Thereby, from Lemma \ref{lem: positivity} and Theorem \ref{thm:eigenvalue inequality},
one has 
\begin{equation}
\mu_{d+d^{\prime}+1}(\varPhi(d+d^{\prime},1)^{\top}\varPhi(d+d^{\prime},1))<1.
\end{equation}
Then $\text{{\bf lim}}_{k\rightarrow\infty}\parallel\boldsymbol{\omega}(k)\parallel=0$,
implying that the multi-agent system \eqref{equ:matrix-consensus-overall}
using Algorithm \ref{MWN-PPAC algorithm} admits average consensus.
\end{proof}
Then, we have shown that the asymptotic average consensus of multi-agent
system \eqref{equ:matrix-consensus-overall} can be guaranteed using
the profile of matrix-valued edge weight generated from Algorithm
\ref{MWN-PPAC algorithm}. We shall proceed to examine the privacy-preservation
performance of the Algorithm \ref{MWN-PPAC algorithm} in the subsequent
discussions.

\section{Privacy-preserving Analysis\label{sec:Privacy}}

In this section, we shall discuss the privacy-preservation performance
of Algorithm \ref{MWN-PPAC algorithm}. We first provide the main
result.
\begin{thm}
\label{thm:Privacy-preserving}By implementing Algorithm \ref{MWN-PPAC algorithm},
for an agent $b\in\mathcal{V}_{\mathcal{B}}$, the privacy of $\boldsymbol{x}_{b}(0)$
is preserved if $b$ has at least one legitimate neighbor $m\in\mathcal{V}_{\mathcal{B}}$.
\end{thm}
\begin{proof}
We shall prove the result via the indistinguishability of a private
value\textquoteright s arbitrary variation to the honest-but-curious
agents. Without loss of generality, assume that agent $b$ has only
one legitimate neighbor $m\in\mathcal{V}_{\mathcal{B}}$. Here, there
are two cases, namely, $(a,m)\in\mathcal{E}$ (Case 1) and $(a,m)\not\in\mathcal{E}$
(Case 2) as shown in Figure \ref{fig:alex-eve-bob}.
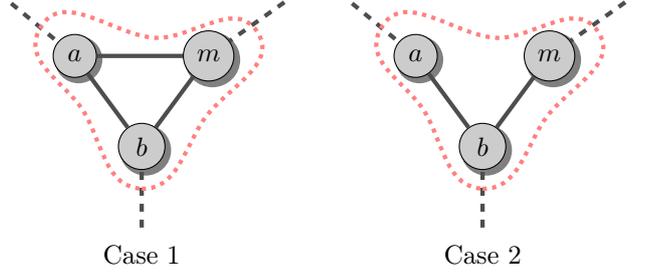
\begin{figure}[H]
\begin{centering}
\begin{tikzpicture}[scale=0.8,use Hobby shortcut]
	\node (ni) at (-1.1,1.5) [circle,circular drop shadow,fill=black!20,draw] {$a$};
    \node (nj) at (0,0) [circle,circular drop shadow,fill=black!20,draw] {$b$};
    \node (nm) at (1.1,1.5) [circle,circular drop shadow,fill=black!20,draw] {$m$};

	\node (niv) at (-2.3,2.5) {};
	\node (njv) at (0,-1.5) {};
	\node (nmv) at (2.5,2.5) {};

	\node (case1) at (0,-1.8) {Case 1};

	\draw[-, ultra thick, color=black!70] (ni) -- (nj); 
	\draw[-, ultra thick, color=black!70] (nj) -- (nm); 
    \draw[-, ultra thick, color=black!70] (ni) -- (nm); 

	\draw[-, dashed, ultra thick, color=black!70] (niv) -- (ni); 
	\draw[-, dashed, ultra thick, color=black!70] (njv) -- (nj); 
    \draw[-, dashed, ultra thick, color=black!70] (nmv) -- (nm);

 \path
  (0,-0.7) coordinate (zj)
  (1.1,0.55) coordinate (zjm)
  (1.9,1.9) coordinate (zm)
  (0.3,1.8) coordinate (zmi)
  (-1.7,1.9) coordinate (zi)
  (-1,0.65) coordinate (zij);
  \draw[color=red!50, ultra thick, dotted, closed] (zj) .. (zjm) .. (zm) .. (zmi) .. (zi) .. (zij);

\end{tikzpicture}\,\,\,\,\,\,\,\begin{tikzpicture}[scale=0.8,use Hobby shortcut]
	\node (ni) at (-1.1,1.5) [circle,circular drop shadow,fill=black!20,draw] {$a$};
    \node (nj) at (0,0) [circle,circular drop shadow,fill=black!20,draw] {$b$};
    \node (nm) at (1.1,1.5) [circle,circular drop shadow,fill=black!20,draw] {$m$};

	\node (niv) at (-2.3,2.5) {};
	\node (njv) at (0,-1.5) {};
	\node (nmv) at (2.5,2.5) {};

	\node (case2) at (0,-1.8) {Case 2};

	\draw[-, ultra thick, color=black!70] (ni) -- (nj); 
	\draw[-, ultra thick, color=black!70] (nj) -- (nm); 

	\draw[-, dashed, ultra thick, color=black!70] (niv) -- (ni); 
	\draw[-, dashed, ultra thick, color=black!70] (njv) -- (nj); 
    \draw[-, dashed, ultra thick, color=black!70] (nmv) -- (nm);

 \path
  (0,-0.7) coordinate (zj)
  (1.1,0.55) coordinate (zjm)
  (1.9,1.9) coordinate (zm)
  (0.3,1.8) coordinate (zmi)
  (-1.7,1.9) coordinate (zi)
  (-1,0.65) coordinate (zij);
  \draw[color=red!50, ultra thick, dotted, closed] (zj) .. (zjm) .. (zm) .. (zmi) .. (zi) .. (zij);

\end{tikzpicture}
\par\end{centering}
\caption{The local communication structures amongst a legitimate agent $b$
and its honest-but-curious neighbor $a$ and legitimate neighbor $m$.}
\label{fig:alex-eve-bob}
\end{figure}

Case 1: We need to show that the privacy of the initial value $\boldsymbol{x}_{b}^{[d^{\prime}+1:d+d^{\prime}]}(0)$
can be preserved against any honest-but-curious agent $a\in\mathcal{V_{\mathcal{A}}}$,
namely, agent $a$ cannot infer the exact value of $\boldsymbol{x}_{b}^{[d^{\prime}+1:d+d^{\prime}]}(0)$
via $\mathcal{O}_{a}^{\mathcal{M}}\left(\{\boldsymbol{x}_{i}(0),\mathcal{I}_{i}\}_{i\in\mathcal{V}}\right)$.
According to Definition \ref{def:definition of privacy}, it suffices
to show that for arbitrary initial value of real state $\bar{\boldsymbol{x}}_{b}^{[d^{\prime}+1:d+d^{\prime}]}(0)\neq\boldsymbol{x}_{b}^{[d^{\prime}+1:d+d^{\prime}]}(0),$
there exist initial values of agent $m$ and associated matrix-valued
edge weights such that the equality \eqref{eq:privacy-definition}
holds, and the agents' states still converge to the original average
value $\text{Avg}(\boldsymbol{x}(0))$ even if $\boldsymbol{x}_{b}(0)$
is changed into $\bar{\boldsymbol{x}}_{b}(0)$. We shall present the
selection of initial value of $\boldsymbol{x}_{m}(0)$ and associated
matrix-valued coupling weights, respectively.

We first examine the selection of initial value of $\boldsymbol{x}_{m}(0)$.
First, we choose $\bar{\boldsymbol{x}}_{m}(0)$ as, 
\begin{equation}
\bar{\boldsymbol{x}}_{m}(0)=\boldsymbol{x}_{m}(0)+\boldsymbol{x}_{b}(0)-\bar{\boldsymbol{x}}_{b}(0).\label{eq:aug-state-condition-0}
\end{equation}
Second, the virtual state of $\bar{\boldsymbol{x}}_{b}(0)$ satisfies
\begin{equation}
\left(\boldsymbol{v}_{1}^{[1:d^{\prime}]}\right)^{\top}\left(\bar{\boldsymbol{x}}_{b}^{[1:d^{\prime}]}(0)-\boldsymbol{x}_{b}^{[1:d^{\prime}]}(0)\right)=0,\label{eq:aug-state-condition-1}
\end{equation}
and
\begin{eqnarray}
 &  & \left(\boldsymbol{v}_{d+d^{\prime}}^{[1:d^{\prime}]}\right)^{\top}\left(\bar{\boldsymbol{x}}_{b}^{[1:d^{\prime}]}(0)-\boldsymbol{x}_{b}^{[1:d^{\prime}]}(0)\right)\nonumber \\
 & = & -\left(\boldsymbol{v}_{d+d^{\prime}}^{[d^{\prime}+1:d^{\prime}+d]}\right)^{\top}(\bar{\boldsymbol{x}}_{b}^{[d^{\prime}+1:d^{\prime}+d]}(0)\nonumber \\
 &  & -\boldsymbol{x}_{b}^{[d^{\prime}+1:d^{\prime}+d]}(0)),\label{eq:aug-state-condition-2}
\end{eqnarray}
and $\boldsymbol{x}_{b}^{[1:d^{\prime}]}(0)-\bar{\boldsymbol{x}}_{b}^{[1:d^{\prime}]}(0)$
and $\boldsymbol{x}_{m}^{[1:d^{\prime}]}(0)-\boldsymbol{x}_{b}^{[1:d^{\prime}]}(0)$
are linearly independent. According to 2) and 3) of Step 2 in $\mathsection$\ref{subsec:Dynamic-Matrix-valued-Edge},
we notice that the initial value subject to the aforementioned constraints
\eqref{eq:aug-state-condition-0}, \eqref{eq:aug-state-condition-1}
and \eqref{eq:aug-state-condition-2} are feasible. 

Finally, we choose
\begin{equation}
\bar{\boldsymbol{x}}_{p}(0)=\boldsymbol{x}_{p}(0),\label{eq:state-condition-for-other-agents}
\end{equation}
for all agent $p\in\mathcal{V}\setminus\left\{ b,m\right\} $. 

Subsequently, we proceed to examine the selection of associated matrix-valued
coupling weights as follows. For $k=0$, we first choose $\bar{A}_{bm}(0)$
such that
\begin{equation}
\bar{A}_{bm}(0)\left(\boldsymbol{x}_{m}(0)-\boldsymbol{x}_{b}(0)\right)=A_{bm}(0)\left(\boldsymbol{x}_{m}(0)-\boldsymbol{x}_{b}(0)\right),\label{eq:weight-matrix-between-friends-1}
\end{equation}
and
\begin{equation}
\bar{A}_{bm}(0)\left(\boldsymbol{x}_{b}(0)-\bar{\boldsymbol{x}}_{b}(0)\right)=\frac{1}{2\sigma}\left(\boldsymbol{x}_{b}(0)-\bar{\boldsymbol{x}}_{b}(0)\right).\label{eq:weight-matrix-between-friends-2}
\end{equation}
One can see that there are $(d+d^{\prime})^{2}$ free variabes in
matrix $\bar{A}_{bm}(0)$ with $2(d+d^{\prime})$ equations in the
equalities \eqref{eq:weight-matrix-between-friends-1} and \eqref{eq:weight-matrix-between-friends-2},
also we note that $\boldsymbol{x}_{b}^{[1:d^{\prime}]}(0)-\bar{\boldsymbol{x}}_{b}^{[1:d^{\prime}]}(0)$
and $\boldsymbol{x}_{m}^{[1:d^{\prime}]}(0)-\boldsymbol{x}_{b}^{[1:d^{\prime}]}(0)$
are linearly independent. Therefore, there exists a matrix $\bar{A}_{bm}(0)\in\mathbb{R}^{n\times n}$
satisfying \eqref{eq:weight-matrix-between-friends-1} and \eqref{eq:weight-matrix-between-friends-2}
simultaneously. 

Second, let $\mathcal{E}^{*}$ denote edges between agents $m$ and
$p$ where $p\in\mathcal{V}_{\mathcal{B}}\setminus\left\{ b,m\right\} $.
We choose $\bar{A}_{pm}(0)$ such that,
\begin{equation}
\bar{A}_{pm}(0)\left(\bar{\boldsymbol{x}}_{m}(0)-\boldsymbol{x}_{p}(0)\right)=A_{pm}(0)\left(\boldsymbol{x}_{m}(0)-\boldsymbol{x}_{p}(0)\right),\label{eq:weight-matrix-between-friends-3}
\end{equation}
for all edges in $\mathcal{E}^{*}$. One can see that there are $(d+d^{\prime})^{2}$
free variabes in matrix $\bar{A}_{pm}(0)$ with $d+d^{\prime}$ equations
in the equalities \eqref{eq:weight-matrix-between-friends-3}, therefore,
there exist a matrix $\bar{A}_{pm}(0)$ satisfying \eqref{eq:weight-matrix-between-friends-3}.

Third, we choose 
\begin{equation}
\bar{A}_{pq}(0)=A_{pq}(0),\label{eq:weight-matrix-between-othet-agents}
\end{equation}
for all $(p,q)\in\mathcal{E}\setminus\left(\left\{ (b,m)\right\} \cup\mathcal{E}^{*}\right)$.

Finally, for $k\ge1$, we choose $\bar{A}_{pq}(k)=A_{pq}(k)$ for
all $(p,q)\in\mathcal{E}$.

From \eqref{eq:aug-state-condition-0} to \eqref{eq:aug-state-condition-2}
and \eqref{eq:weight-matrix-between-othet-agents}, one can conclude
that,
\begin{equation}
\bar{A}_{ab}(0)\bar{\boldsymbol{x}}_{b}(0)=A_{ab}(0)\boldsymbol{x}_{b}(0)\label{eq:adversary-friend-1}
\end{equation}
 and 
\begin{equation}
\bar{A}_{am}(0)\bar{\boldsymbol{x}}_{m}(0)=A_{am}(0)\boldsymbol{x}_{m}(0),\label{eq:adversary-friend-2}
\end{equation}
for all $a\in\mathcal{V}_{\mathcal{A}}$.

To sum up the above analysis, one can verify that $\bar{\boldsymbol{x}}_{q}(1)=\boldsymbol{x}_{q}(1)$
for all $q\in\mathcal{V}$. Specifically, according to \eqref{eq:state-condition-for-other-agents},
\eqref{eq:weight-matrix-between-othet-agents}, \eqref{eq:adversary-friend-1}
and \eqref{eq:adversary-friend-2}, one has, 
\begin{eqnarray*}
 &  & \bar{\boldsymbol{x}}_{a}(1)\\
 & = & \bar{\boldsymbol{x}}_{a}(0)+\sigma\bar{A}_{ab}(0)\bar{\boldsymbol{x}}_{b}(0)-\sigma\bar{A}_{ab}(0)\bar{\boldsymbol{x}}_{a}(0)\\
 & + & \sigma\bar{A}_{am}(0)\bar{\boldsymbol{x}}_{m}(0)-\sigma\bar{A}_{am}(0)\bar{\boldsymbol{x}}_{a}(0)\\
 & + & \sigma{\displaystyle \sum_{p\in\mathcal{N}_{a}\setminus\{b,m\}}}\left(\bar{A}_{ap}(0)\bar{\boldsymbol{x}}_{p}(0)-\bar{A}_{ap}(0)\bar{\boldsymbol{x}}_{a}(0)\right)\\
 & = & \boldsymbol{x}_{a}(0)+\sigma A_{ab}(0)\boldsymbol{x}_{b}(0)-\sigma A_{ab}(0)\boldsymbol{x}_{a}(0)\\
 & + & \sigma A_{am}(0)\boldsymbol{x}_{m}(0)-\sigma A_{am}(0)\boldsymbol{x}_{a}(0)\\
 & + & \sigma{\displaystyle \sum_{p\in\mathcal{N}_{a}\setminus\{b,m\}}}\left(A_{ap}(0)\boldsymbol{x}_{p}(0)-A_{ap}(0)\boldsymbol{x}_{a}(0)\right)\\
 & = & \boldsymbol{x}_{a}(1),
\end{eqnarray*}
for all $a\in\mathcal{V}_{\mathcal{A}}$. 

According to \eqref{eq:state-condition-for-other-agents}, \eqref{eq:weight-matrix-between-friends-1},
\eqref{eq:weight-matrix-between-friends-2} and \eqref{eq:adversary-friend-1},
one has,

\begin{eqnarray*}
 &  & \bar{\boldsymbol{x}}_{b}(1)\\
 & = & \bar{\boldsymbol{x}}_{b}(0)+\sigma{\displaystyle \sum_{p\in\mathcal{V}\setminus\left\{ b,m\right\} }}\left(\bar{A}_{bp}(0)\bar{\boldsymbol{x}}_{p}(0)-\bar{A}_{bp}(0)\bar{\boldsymbol{x}}_{b}(0)\right)\\
 & + & \sigma\bar{A}_{bm}(0)\bar{\boldsymbol{x}}_{m}(0)-\sigma\bar{A}_{bm}(0)\bar{\boldsymbol{x}}_{b}(0)\\
 & = & \bar{\boldsymbol{x}}_{b}(0)+\sigma{\displaystyle \sum_{p\in\mathcal{V}\setminus\left\{ b,m\right\} }}\left(A_{bp}(0)\boldsymbol{x}_{p}(0)-A_{bp}(0)\boldsymbol{x}_{b}(0)\right)\\
 & + & \sigma\bar{A}_{bm}(0)\left(\boldsymbol{x}_{m}(0)+\boldsymbol{x}_{b}(0)-\bar{\boldsymbol{x}}_{b}(0)-\bar{\boldsymbol{x}}_{b}(0)\right)\\
 & = & \bar{\boldsymbol{x}}_{b}(0)+\sigma{\displaystyle \sum_{p\in\mathcal{V}\setminus\left\{ b,m\right\} }}\left(A_{bp}(0)\boldsymbol{x}_{p}(0)-A_{bp}(0)\boldsymbol{x}_{b}(0)\right)\\
 & + & \sigma\bar{A}_{bm}(0)\left(\boldsymbol{x}_{m}(0)-\boldsymbol{x}_{b}(0)+2\boldsymbol{x}_{b}(0)-2\bar{\boldsymbol{x}}_{b}(0)\right)\\
 & = & \bar{\boldsymbol{x}}_{b}(0)+\sigma{\displaystyle \sum_{p\in\mathcal{V}\setminus\left\{ b,m\right\} }}\left(A_{bp}(0)\boldsymbol{x}_{p}(0)-A_{bp}(0)\boldsymbol{x}_{b}(0)\right)\\
 & + & \sigma\bar{A}_{bm}(0)\left(\boldsymbol{x}_{m}(0)-\boldsymbol{x}_{b}(0)\right)\\
 & + & \sigma\bar{A}_{bm}(0)\left(2\boldsymbol{x}_{b}(0)-2\bar{\boldsymbol{x}}_{b}(0)\right)\\
 & = & \boldsymbol{x}_{b}(0)+\sigma A_{bm}(0)\left(\boldsymbol{x}_{m}(0)-\boldsymbol{x}_{b}(0)\right)\\
 &  & \sigma{\displaystyle \sum_{p\in\mathcal{V}\setminus\left\{ b,m\right\} }}\left(A_{bp}(0)\boldsymbol{x}_{p}(0)-A_{bp}(0)\boldsymbol{x}_{b}(0)\right)\\
 & = & \boldsymbol{x}_{b}(1).
\end{eqnarray*}
And, by \eqref{eq:state-condition-for-other-agents} to \eqref{eq:weight-matrix-between-friends-3}
and \eqref{eq:adversary-friend-2}, one has,
\begin{eqnarray*}
 &  & \bar{\boldsymbol{x}}_{m}(1)\\
 & = & \bar{\boldsymbol{x}}_{m}(0)+\sigma\bar{A}_{mb}(0)\bar{\boldsymbol{x}}_{b}(0)-\sigma\bar{A}_{mb}\bar{\boldsymbol{x}}_{m}(0)\\
 & + & \sigma{\displaystyle \sum_{p\in\mathcal{V}\setminus\left\{ b,m\right\} }}\left(\bar{A}_{mp}(0)\bar{\boldsymbol{x}}_{p}(0)-\bar{A}_{mp}(0)\bar{\boldsymbol{x}}_{m}(0)\right)\\
 & = & \bar{\boldsymbol{x}}_{m}(0)+\sigma\bar{A}_{mb}(0)\left(\bar{\boldsymbol{x}}_{b}(0)-\boldsymbol{x}_{m}(0)-\boldsymbol{x}_{b}(0)+\bar{\boldsymbol{x}}_{b}(0)\right)\\
 & + & \sigma{\displaystyle \sum_{p\in\mathcal{V}\setminus\left\{ b,m\right\} }}\left(A_{mp}(0)\boldsymbol{x}_{p}(0)-A_{mp}(0)\boldsymbol{x}_{m}(0)\right)\\
 & = & \bar{\boldsymbol{x}}_{m}(0)+\sigma\bar{A}_{mb}(0)\left(\boldsymbol{x}_{b}(0)-\boldsymbol{x}_{m}(0)\right)\\
 & + & \sigma\bar{A}_{mb}(0)\left(2\bar{\boldsymbol{x}}_{b}(0)-2\boldsymbol{x}_{b}(0)\right)\\
 & + & \sigma{\displaystyle \sum_{p\in\mathcal{V}\setminus\left\{ b,m\right\} }}\left(A_{mp}(0)\boldsymbol{x}_{p}(0)-A_{mp}(0)\boldsymbol{x}_{m}(0)\right)\\
 & = & \boldsymbol{x}_{m}(0)+\boldsymbol{x}_{b}(0)-\bar{\boldsymbol{x}}_{b}(0)\\
 & + & \sigma A_{mb}(0)\left(\boldsymbol{x}_{b}(0)-\boldsymbol{x}_{m}(0)\right)+\left(\bar{\boldsymbol{x}}_{b}(0)-\boldsymbol{x}_{b}(0)\right)\\
 & + & \sigma{\displaystyle \sum_{p\in\mathcal{V}\setminus\left\{ b,m\right\} }}\left(A_{mp}(0)\boldsymbol{x}_{p}(0)-A_{mp}(0)\boldsymbol{x}_{m}(0)\right)\\
 & = & \boldsymbol{x}_{m}(0)+\sigma A_{mb}(0)\left(\boldsymbol{x}_{b}(0)-\boldsymbol{x}_{m}(0)\right)\\
 & + & \sigma{\displaystyle \sum_{p\in\mathcal{V}\setminus\left\{ b,m\right\} }}\left(A_{mp}(0)\boldsymbol{x}_{p}(0)-A_{mp}(0)\boldsymbol{x}_{m}(0)\right)\\
 & = & \boldsymbol{x}_{m}(1).
\end{eqnarray*}
Finally, for all $p\in\mathcal{V}_{\mathcal{B}}\setminus\left\{ b,m\right\} $,
according to \eqref{eq:state-condition-for-other-agents}, \eqref{eq:weight-matrix-between-friends-3}
and \eqref{eq:weight-matrix-between-othet-agents}, one has,

\begin{eqnarray*}
 &  & \bar{\boldsymbol{x}}_{p}(1)\\
 & = & \bar{\boldsymbol{x}}_{p}(0)+\sigma{\displaystyle \sum_{q\in\mathcal{N}_{p}}}\left(\bar{A}_{pq}(0)\bar{\boldsymbol{x}}_{q}(0)-\bar{A}_{pq}(0)\bar{\boldsymbol{x}}_{p}(0)\right)\\
 & = & \boldsymbol{x}_{p}(0)+\sigma{\displaystyle \sum_{q\in\mathcal{N}_{p}}}\left(A_{pq}(0)\boldsymbol{x}_{q}(0)-A_{pq}(0)\boldsymbol{x}_{p}(0)\right)\\
 & = & \boldsymbol{x}_{p}(1).
\end{eqnarray*}

We notice that $\bar{A}_{pq}(k)=A_{pq}(k)$ for all $(p,q)\in\mathcal{E}$
and $k\ge1$; As such \eqref{eq:privacy-definition} in Definition
\ref{def:definition of privacy} can be straightforward verified.

Case 2: The proof of this case is similar in spirit with the Case
1, hence omitted for brevity. 
\end{proof}
\begin{rem}
Consider the multi-agent system \eqref{equ:matrix-consensus-overall}
adopting Algorithm 1. If an agent $b$ is connected to the remaining
agents in the network only through an (or a group of colluding) honest-but-curious
agent $a$, then the initial state of agent $b$ can be uniquely inferred
by agent $a$ in an asymptotic sense. To see this, we notice that
\begin{align}
\boldsymbol{x}_{b}(0) & =\boldsymbol{x}_{b}(l)\nonumber \\
 & -\sigma\sum_{k=0}^{l-1}A_{ab}(k)\boldsymbol{x}_{a}(k)+\sigma\sum_{k=0}^{l-1}A_{ab}(k)\boldsymbol{x}_{b}(k).\label{eq:inferring initial state}
\end{align}
The average consensus is achieved as $l$ goes to infinity, namely,
i.e., $\mathbf{lim}_{l\rightarrow\infty}\boldsymbol{x}_{a}(l)=\mathbf{lim}_{l\rightarrow\infty}\boldsymbol{x}_{b}(l)$.
Therefore, agent $a$ can uniquely infer the initial state of agent
$b$ through \eqref{eq:inferring initial state}.
\end{rem}
\begin{rem}
For $k\geq1$, recall that the information of agent $i$ derived from
agent $j$ is $\boldsymbol{y}_{j\rightarrow i}(k)=A_{ij}(k)\boldsymbol{x}_{j}(k)$;
then one has

\[
V_{\rho(k)}\text{{\bf diag}}\left(\gamma_{ij}^{\rho(k)},\zeta_{ij}^{\rho(k)},0,\cdots,0\right)V_{\rho(k)}^{\top}\boldsymbol{x}_{j}(k)=\boldsymbol{y}_{j\rightarrow i}(k),
\]
where 
\[
V_{\rho(k)}=\left(\frac{\boldsymbol{v}_{\rho(k)}}{\parallel\boldsymbol{v}_{\rho(k)}\parallel},\frac{\boldsymbol{v}_{d+d^{\prime}}}{\parallel\boldsymbol{v}_{d+d^{\prime}}\parallel},P_{0}\right).
\]
Thus, 
\begin{align*}
 & \boldsymbol{v}_{\rho(k)}^{\top}\boldsymbol{x}_{j}(k)\\
= & \frac{1}{\gamma_{ij}^{\rho(k)}}\left[\left(\frac{\boldsymbol{v}_{\rho(k)}}{\parallel\boldsymbol{v}_{\rho(k)}\parallel},\frac{\boldsymbol{v}_{d+d^{\prime}}}{\parallel\boldsymbol{v}_{d+d^{\prime}}\parallel},P_{0}\right)^{\top}\boldsymbol{y}_{j\rightarrow i}(k)\right]_{1},
\end{align*}
 and 
\begin{align*}
 & \boldsymbol{v}_{d+d^{\prime}}^{\top}\boldsymbol{x}_{j}(k)\\
 & =\frac{1}{\zeta_{ij}^{\rho(k)}}\left[\left(\frac{\boldsymbol{v}_{\rho(k)}}{\parallel\boldsymbol{v}_{\rho(k)}\parallel},\frac{\boldsymbol{v}_{d+d^{\prime}}}{\parallel\boldsymbol{v}_{d+d^{\prime}}\parallel},P_{0}\right)^{\top}\boldsymbol{y}_{j\rightarrow i}(k)\right]_{2}.
\end{align*}
We notice that there exists $d+d^{\prime}$ free variables in at most
two equations, therefore agent $i$ can not infer $\boldsymbol{x}_{j}(k)$
associated with agent $j$ for $k\geq1$.
\end{rem}
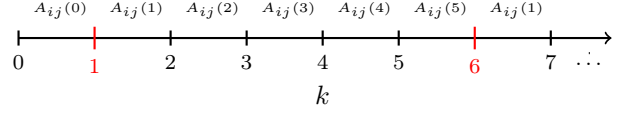
\begin{figure}[H]
\begin{centering}
\begin{tikzpicture}

\draw[->,thick] (0,0) -- (7.8,0) node[anchor=north] {};

\draw	(4,-0.5) node[anchor=north] {$k$};

\draw[-,thick] (0,0.1) -- (0,-0.1) node[anchor=north,font=\fontsize{8}{8}\selectfont] {0};

\draw[-,thick,color=red] (1,0.15) -- (1,-0.15) node[anchor=north,font=\fontsize{8}{8}\selectfont] {1};

\draw[-,thick] (2,0.1) -- (2,-0.1) node[anchor=north,font=\fontsize{8}{8}\selectfont] {2};

\draw[-,thick] (3,0.1) -- (3,-0.1) node[anchor=north,font=\fontsize{8}{8}\selectfont] {3};

\draw[-,thick] (4,0.1) -- (4,-0.1) node[anchor=north,font=\fontsize{8}{8}\selectfont] {4};

\draw[-,thick] (5,0.1) -- (5,-0.1) node[anchor=north,font=\fontsize{8}{8}\selectfont] {5};

\draw[-,thick,color=red] (6,0.15) -- (6,-0.15) node[anchor=north,font=\fontsize{8}{8}\selectfont] {6};

\draw[-,thick] (7,0.1) -- (7,-0.1) node[anchor=north,font=\fontsize{8}{8}\selectfont] {7};

\draw[-,thick] (7.5,0.-0.2) -- (7.5,-0.2) node[anchor=north,font=\fontsize{8}{8}\selectfont] {$\ldots$};

\draw	
		 (0.55,0.6) node[anchor=north] {{\tiny $A_{ij}(0)$}}
		 (1.55,0.6) node[anchor=north] {{\tiny $A_{ij}(1)$}}
		 (2.55,0.6) node[anchor=north] {{\tiny $A_{ij}(2)$}}
		 (3.55,0.6) node[anchor=north] {{\tiny $A_{ij}(3)$}}
		 (4.55,0.6) node[anchor=north] {{\tiny $A_{ij}(4)$}}
		 (5.55,0.6) node[anchor=north] {{\tiny $A_{ij}(5)$}}
		 (6.55,0.6) node[anchor=north] {{\tiny $A_{ij}(1)$}}
;

\end{tikzpicture}
\par\end{centering}
\caption{An illustration for the switching pattern of the matrix-valued weights
on edge $(i,j)\in\mathcal{E}$ with period $T=5$. The boundaries
between neighboring period are highlighted by vertical red bars.}
\label{fig:switching-pattern}
\end{figure}

\section{Simulation\label{sec:Simulation}}

We now provide a simulation example to illustrate the effectiveness
of the proposed MWN-PPAC algorithm. Consider the $5$-node network
in Figure \ref{fig:integral-network} where each agent holds a $3$-dimensional
state. Choose the dimension of virtual state $d^{\prime}=3$ and the
matrix-valued weight on each edge $(i,j)\in\mathcal{E}$ adopting
a periodic switching pattern with period $T=5$, as illustrated in
Figure \ref{fig:switching-pattern}. We choose the construction of
the orthogonal vector set $\mathfrak{V}$ as the same as that discussed
in Remark \ref{rem:feasibility of B}, namely, by \eqref{eq:v1} to
\eqref{eq:v6}. We choose the initial states 
\[
\boldsymbol{x}^{\top}(0)=(\boldsymbol{x}_{1}^{\top}(0),\thinspace\boldsymbol{x}_{2}^{\top}(0),\thinspace\boldsymbol{x}_{3}^{\top}(0),\thinspace\boldsymbol{x}_{4}^{\top}(0),\thinspace\boldsymbol{x}_{5}^{\top}(0))^{\top},
\]
as follows,
\[
\boldsymbol{x}_{1}(0)=(0.20~~0.30~~0.25~~0.60~~0.32~~0.65)^{\top},
\]
\[
\boldsymbol{x}_{2}(0)=(0.60~~0.72~~0.57~~0.24~~0.91~~0.95)^{\top},
\]
\[
\boldsymbol{x}_{3}(0)=(0.52~~0.71~~0.80~~0.20~~0.12~~0.62)^{\top},
\]
\[
\boldsymbol{x}_{4}(0)=(0.02~~0.04~~0.12~~0.82~~0.38~~0.23)^{\top},
\]
\[
\boldsymbol{x}_{5}(0)=(0.37~~0.17~~0.77~~0.33~~0.32~~0.72)^{\top}.
\]
Note that the virtual state can be randomly chosen. Choose $\sigma=2$
and randomly choose $\gamma_{ij}^{\rho(k)},\zeta_{ij}^{\rho(k)}\in(0,\frac{1}{4(n-1)\sigma})$,
$\alpha_{ij}>0$, $\beta_{ij}>0$ for $(i,j)\in\mathcal{E}$ and $k>0$.
The average value of $\boldsymbol{x}(0)$ is 
\begin{align*}
\text{Avg}(\boldsymbol{x}(0)) & =\frac{1}{5}{\displaystyle \sum_{j=1}^{5}}\boldsymbol{x}_{j}(0)\\
 & =(0.34~~0.39~~0.50~~0.44~~0.41~~0.63)^{\top}.
\end{align*}

As one can see, average consensus is asymptotically achieved; see
Figure \ref{fig:evolution-real-state}.

\begin{figure}
\begin{centering}
\includegraphics[width=9cm]{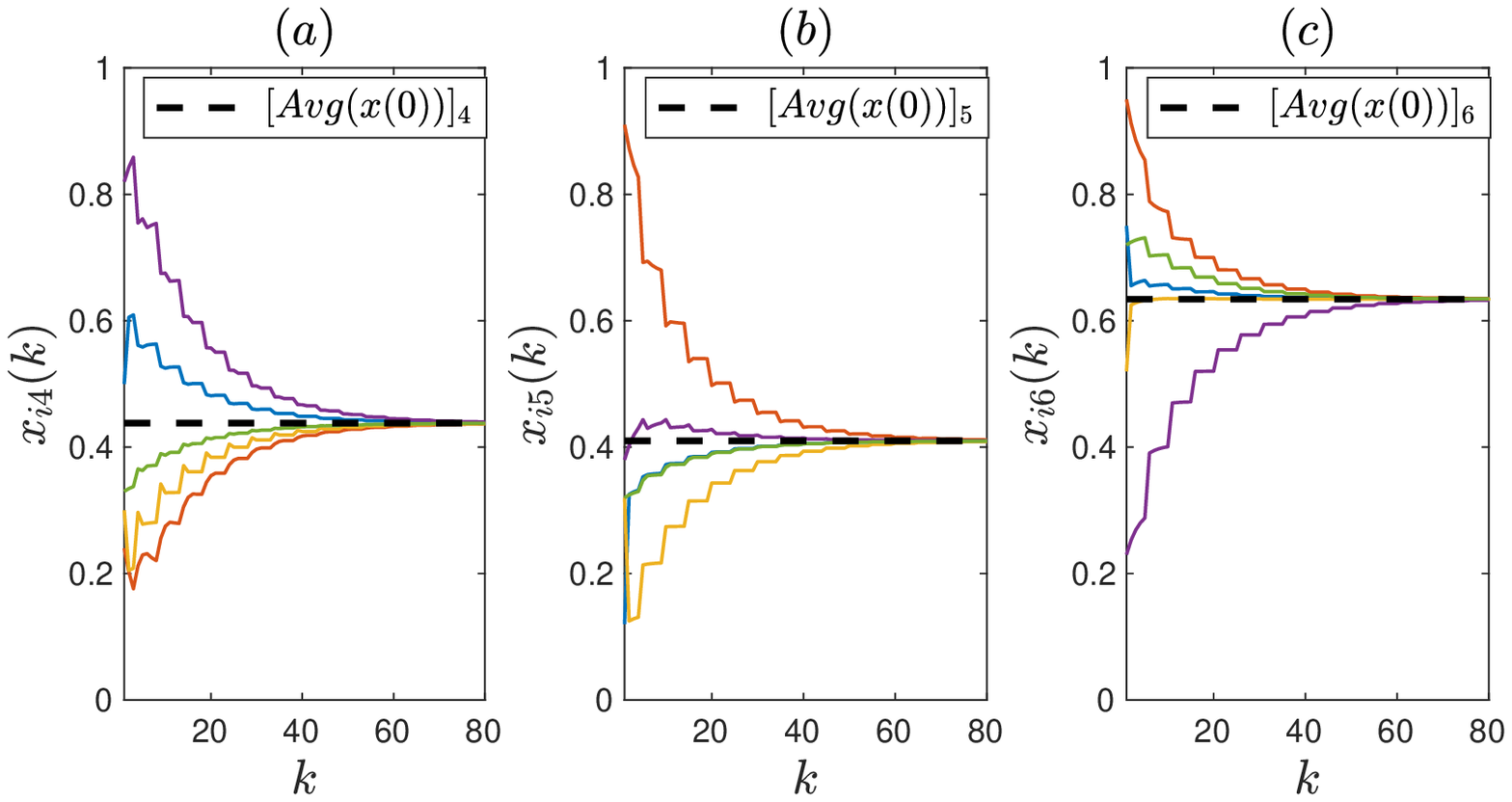}
\par\end{centering}
\caption{Evolution of real agent states $\boldsymbol{x}_{i4}(k),\boldsymbol{x}_{i5}(k)$
and $\boldsymbol{x}_{i6}(k)$.}
\label{fig:evolution-real-state}
\end{figure}

Moreover, consider privacy preservation on the information channel
from agent $3$ to its neighboring agent $2$. The transmitted information
versus the state of the agent are shown in Figure \ref{fig:transmited-real}.
One can see that the information exchanged amongst neighbor agents
is not related to the true state of the corresponding agents. 

\begin{figure}
\begin{centering}
\includegraphics[width=9cm]{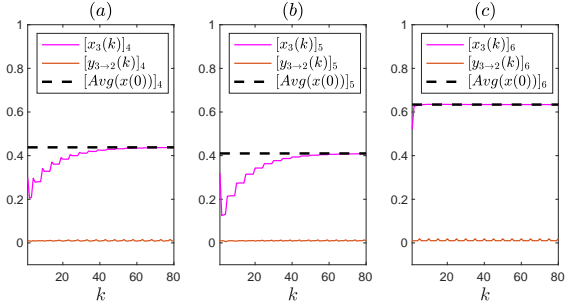}
\par\end{centering}
\caption{Trajectories of real agent states $[\boldsymbol{x}_{3}(k)]_{4}$,
$[\boldsymbol{x}_{3}(k)]_{5}$ and $[\boldsymbol{x}_{3}(k)]_{6}$
versus the information transmited from agent 3 to agent 2, namely,
$[y_{3\rightarrow2}(k)]_{4}$, $[y_{3\rightarrow2}(k)]_{5}$ and $[y_{3\rightarrow2}(k)]_{6}$.}
\label{fig:transmited-real}
\end{figure}

We further examine another arbitrary initial value, namely,
\[
\bar{\boldsymbol{x}}^{\top}(0)=(\bar{\boldsymbol{x}}_{1}^{\top}(0),\thinspace\bar{\boldsymbol{x}}_{2}^{\top}(0),\thinspace\bar{\boldsymbol{x}}_{3}^{\top}(0),\thinspace\bar{\boldsymbol{x}}_{4}^{\top}(0),\thinspace\bar{\boldsymbol{x}}_{5}^{\top}(0))^{\top},
\]
such that $\bar{\boldsymbol{x}}_{i}(0)=\boldsymbol{x}_{i}(0)$ for
$i\in\left\{ 2,4,5\right\} $ and 
\[
\bar{\boldsymbol{x}}_{1}(0)=(0.30~~0.20~~0.65~~0.50~~0.12~~0.75)^{\top},
\]
\[
\bar{\boldsymbol{x}}_{3}(0)=(0.42~~0.81~~0.40~~0.30~~0.32~~0.52)^{\top}.
\]
A comparison of the state evolution of Algorithm 1 initiated from
$\boldsymbol{x}(0)$ and $\bar{\boldsymbol{x}}(0)$, respectively,
is shown in Figures \ref{fig:trajectory_real_comparison} and \ref{fig:trajectory_virtual_comparison},
where only $k=0,1,2,3,4,5$ are shown. One can see that although Algorithm
1 is initiated from different initial values, the real states in each
dimension coincide at $k=1$ and therefore at all $k>1$.

\begin{figure}
\begin{centering}
\includegraphics[width=9cm]{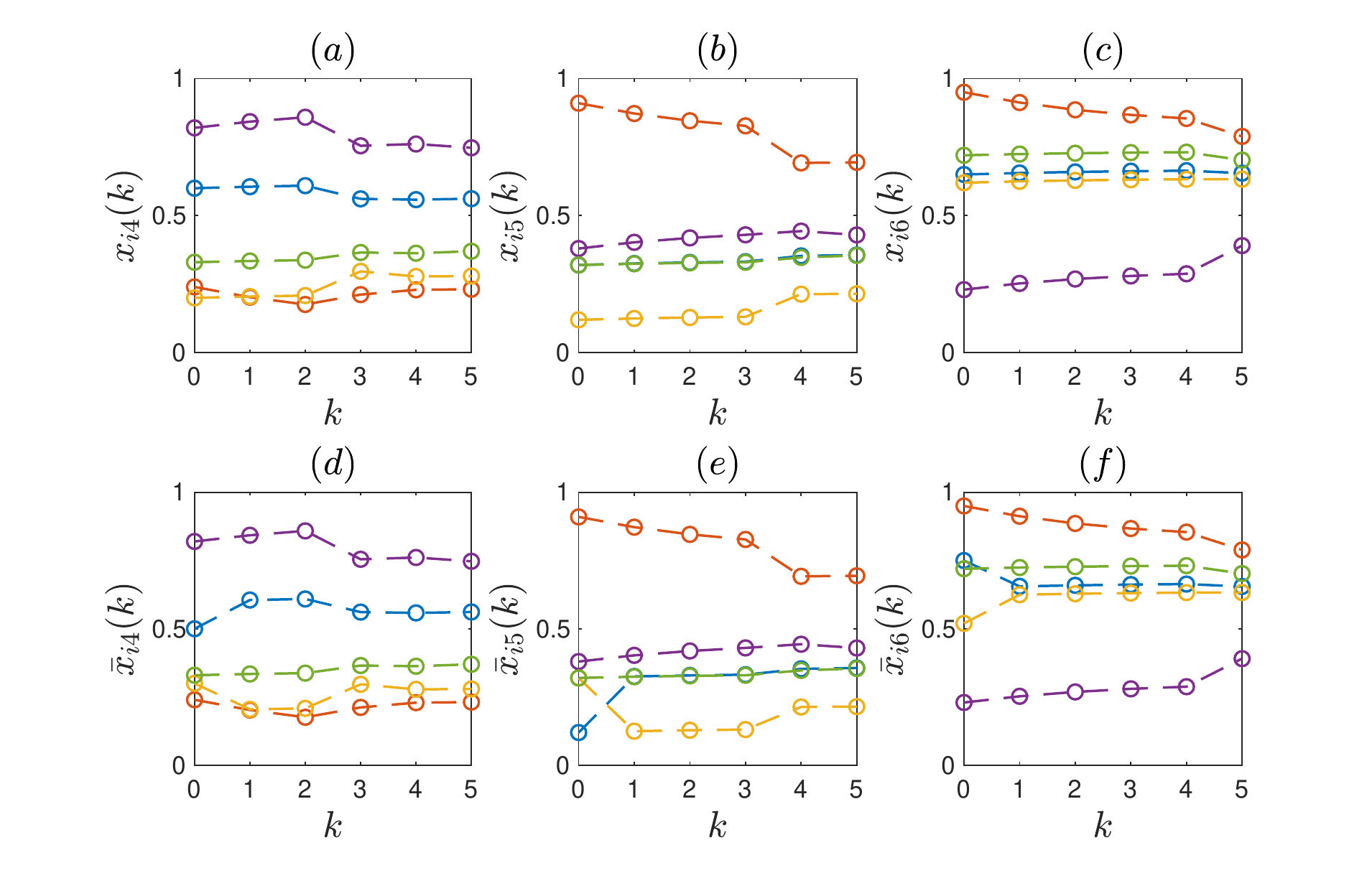}
\par\end{centering}
\caption{Trajectories of real agent states initiated from $\boldsymbol{x}(0)$
(($a$)-($c$)) and that by $\bar{\boldsymbol{x}}(0)$ (($d$)-($f$)),
respectively. }
\label{fig:trajectory_real_comparison}
\end{figure}

\begin{figure}
\begin{centering}
\includegraphics[width=9cm]{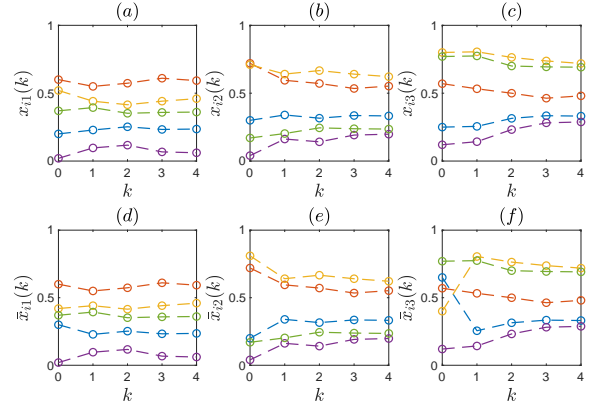}
\par\end{centering}
\caption{Trajectories of virtual agent states initiated from $\boldsymbol{x}(0)$
(($a$)-($c$)) and that by $\bar{\boldsymbol{x}}(0)$ (($d$)-($f$)),
respectively.}
\label{fig:trajectory_virtual_comparison}
\end{figure}

\section{Conclusion Remarks\label{sec:Conclusion}}

In this paper, we proposed an algorithmic framework for vector-valued
privacy-preserving average consensus. The algorithm is essentially
built on two principles, namely, agent state lifting and dynamic matrix-valued
weight design. A self-contained analysis in terms of convergence and
privacy preservation of the algorithm was then provided. The proposed
algorithm is simple and efficient, and can be implemented in a distributed
manner. The proposed algorithm provides a new dimension for PPAC algorithm
design by utilizing the proper lifting of agents' state space. 

Future works in this direction  include examining privacy-preserving
distributed algorithms on matrix-weighted networks in the context
of estimation, control, and optimization.

\bibliographystyle{elsarticle-harv}
\addcontentsline{toc}{section}{\refname}\bibliography{mybib}

\end{document}